\documentclass[11pt]{article}
\usepackage{amssymb}
\usepackage{amsmath}
\usepackage{theorem}
\usepackage{latexsym}
\usepackage{graphicx}
\usepackage{mathtools}
\reversemarginpar
\theoremheaderfont{\bf}
\theorembodyfont{\sl}
\newtheorem{theo}{Theorem}[section]
{\theorembodyfont{\rm} \newtheorem{defi}[theo]{Definition}}
{\theorembodyfont{\rm} \newtheorem{exa}[theo]{Example}}
{\theorembodyfont{\rm} \newtheorem{rem}[theo]{Remark}}

\newtheorem{cor}[theo]{Corollary}
\newtheorem{conj}[theo]{Conjecture}

{\theorembodyfont{\rm} }
\newtheorem{lemma}[theo]{Lemma}
{\theorembodyfont{\rm}}
{\theorembodyfont{\rm}}
\numberwithin{equation}{section}
\newenvironment{proof}{{\sc Proof:}}{\mbox{}\hfill$\Box$\medskip}
\newcommand{\QED}{\mbox{}\hfill$\Box$\medskip}

\newcommand{\junk}[1]{}

\newcommand{\F}{{\mathbb F}}
\newcommand{\Z}{{\mathbb Z}}

\newcommand{\cA}{{\mathcal A}}

\newcommand{\cC}{{\mathcal C}}

\newcommand{\cR}{{\mathcal R}}
\newcommand{\cS}{{\mathcal S}}
\newcommand{\cT}{{\mathcal T}}
\newcommand{\cU}{{\mathcal U}}

\newcommand{\cX}{{\mathcal X}}
\newcommand{\cY}{{\mathcal Y}}
\newcommand{\cZ}{{\mathcal Z}}
\newcommand{\Bf}{{\mathfrak B}}
\newcommand{\jk}{{[j,k)}}
\newcommand{\zerob}{{\mathbf 0}}
\newcommand{\ab}{{\mathbf a}}

\renewcommand{\sb}{{\mathbf s}}
\newcommand{\inner}[1]{\langle{#1}\rangle}

\newcounter{alp}
\newcounter{ara}
\newcounter{rom}
\newenvironment{romanlist}{\begin{list}{(\roman{rom})\hfill}{\usecounter{rom}
     \topsep0ex \labelwidth.7cm \leftmargin.7cm \labelsep0cm
     \rightmargin0cm \parsep0ex \itemsep.3ex
     \partopsep1.6ex}}{\end{list}}

\begin{document}
\title{Local Irreducibility of Tail-Biting Trellises}

\author{Heide Gluesing-Luerssen and G. David Forney, Jr.\footnote{H. Gluesing-Luerssen was partially supported
by the National Science Foundation Grants DMS-0908379 and DMS-1210061.
The authors are with the Department of Mathematics, University of Kentucky, Lexington KY 40506-0027
(email: heide.gl@uky.edu),
and the Laboratory for Information and Decision Systems,
Massachusetts Institute of Technology,
Cambridge, MA 02139 (email: forneyd@comcast.net), respectively.}}
\date{\today}

\vspace*{-5ex}
{\let\newpage\relax\maketitle}

\begin{abstract}
This paper investigates tail-biting trellis realizations for linear block codes.
Intrinsic trellis properties are used to characterize irreducibility on given intervals of
the time axis.
It proves beneficial to always consider the trellis and its dual simultaneously.
A major role is played by trellis properties that amount to
observability and controllability for fragments of the trellis of various lengths.
For fragments of length less than the minimum span length of the code it is shown that
fragment observability and fragment controllability are equivalent to irreducibility.
For reducible trellises, a constructive reduction procedure is presented.
The considerations also lead to a characterization for when the dual of a trellis
allows a product factorization into elementary (``atomic'') trellises.
\end{abstract}

\section{Introduction}\label{S-Intro}
The powerful performance of iterative decoding algorithms for codes on graphs has made graphical models
a major topic in coding theory.
In particular, it has led to a vivid interest in optimal graphical representations of (linear)
block codes.
For cycle-free graphs, this realization theory is by now well understood.
In this case, a realization is minimal if and only if it is trim and proper
(i.e., every state occurs in some constraint codeword and no constraint codeword is supported by
a single state variable), and minimal realizations are unique up to state space isomorphisms.
Moreover, every non-minimal realization can be reduced to a minimal one by a process of trimming and merging.
For details on all of this, see~\cite{FGL12,Ka09}, as well as the excellent survey~\cite{Va98} for the special case of
conventional trellis realizations.

The focus of this paper is on linear tail-biting trellises, which form the simplest type of realizations
on graphs with cycles.
Tail-biting trellises gained a lot of attention after the appearance of~\cite{CFV99}, where it was shown that,
for a given code, the complexity of a tail-biting trellis realization may be considerably lower than that of the
best conventional (i.e., cycle-free) trellis.
This resulted in increased study of (minimal) tail-biting trellises
\cite{LiSh00,ShBe00,Koe02,KoVa03,NoSh06,GLW11,GLW11a,Con12} as well as for normal realization on general
graphs~\cite{Fo01,MaoKsch05,Ka09,AlMao11}.

A systematic theory of tail-biting trellis realizations was initiated by Koetter/Vardy in their landmark paper~\cite{KoVa03}.
Among other things, they highlighted the fundamental problem that all meaningful concepts
of complexity measures for tail-biting trellises lead to different (pre-)orderings, all of which are only partial.
As a consequence, a given code does not have a unique minimal trellis realization.
By extending factorization ideas from conventional trellises, Koetter/Vardy showed in~\cite{KoVa02} that every
reduced trellis (i.e., all states and branches appear on valid trajectories) is a product trellis, that is, it can
can be obtained as the product of elementary trellises.
In~\cite{KoVa03}, they showed that product trellises based on ``shortest generators'' of the code in a circular
interval sense, form a reasonably small class of trellises which is guaranteed to contain all minimal trellises.

In~\cite{GLW11,GLW11a} it is shown that these trellises, called KV-trellises, enjoy nice properties;
for instance, they are non-mergeable, and the dual of a KV-trellis (in the sense of normal realization dualization~\cite{Fo01})
is a KV-trellis.
As a consequence, the dual of a KV-trellis is a product trellis representing the dual code.
The last property is remarkable because in general the dual of a product trellis (even if non-mergeable)
is not a product trellis.
Yet, this invariance under dualization does not characterize KV-trellises, and in fact, no intrinsic characterizations for
being a KV-trellis or being minimal are known.
As a consequence, no constructive method for reducing a given realization to a minimal one has been found yet.

In this paper, we study tail-biting trellises, and make some progress on these aforementioned topics.
We derive constructive procedures to reduce the complexity of a trellis, and derive
irreducibility criteria.

We began our investigation of irreducibility in~\cite{FGL12}, where we considered normal realizations on general graphs.
In that paper, a local reduction was defined as the replacement of one state space by a smaller one, and an adjustment
of the adjacent constraint codes, without changing the rest of the realization or the code that is realized.
We showed that a realization that is not trim, proper, observable and controllable (TPOC) may be locally reduced by a trimming
or merging operation on an appropriate state space.
We note that in the prior literature, starting with~\cite{VaKsch96}, merging has been studied much more than trimming, no
doubt because it seems obvious to trim unused states;  however, because trimming and merging are dual operations~\cite{FGL12},
we weight them equally.

This paper continues our investigation for the special case of linear tail-biting trellises.
In this case, the reductions considered in~\cite{FGL12} may be regarded as
reductions of trellis fragments of length~$2$ (two consecutive constraint codes and the state space involved in
both of them), and it appears natural to extend such reductions to fragments of any length.
As we will see, reducibility is then closely related to properties of the dual trellis realization.
Indeed, if the given trellis or its dual lacks certain basic, easily detectable properties,
then both can be reduced simultaneously.
In this sense, the dual trellis may reveal defects that are not immediately apparent in the primal trellis.
It is therefore natural to treat a trellis and its dual on an equal footing, so that all reductions come
with an analogue for the dual trellis.

In this fashion we can show that the necessary properties for irreducibility on fragments of length~$2$, presented
in~\cite{FGL12}, extend to necessary conditions for irreducibility on longer fragments.
They amount to fragment observability and controllability, and are closely related to trimness of the trellis and its dual
in a fragment sense.
We also prove that these conditions are sufficient for irreducibility on fragments of length less than the minimum
span length of the code, which is a measure of the lengths of zero runs in the codewords.
We then discuss the remaining case of reducibility on longer fragments, and illustrate with an example
how this problem may be approached; however, this case remains largely open.

Finally, we relate our results to the approach taken by Koetter and Vardy in~\cite{KoVa02,KoVa03}.
As mentioned earlier, they investigated the class of reduced trellises and narrowed it further down to KV-trellises in their
search for minimal trellises.
In this paper, we do not assume reducedness because the dual of a reduced trellis is not
necessarily reduced.
In fact, our results will indeed provide an easy-to-check criterion for when the dual of an observable reduced
trellis is a reduced trellis.
We also present a summary of various trellis classes and their relationship.

\section{Codes, Trellises, and Reducibility}\label{S-Basics}
In this section we introduce the basic notions for trellises as needed in this paper.
Most of it is simply a specialization of the terminology used in~\cite{FGL12} for general normal realizations.
In addition, we also define local reductions of trellises; this will be more refined than the definition used in~\cite{FGL12}.

Throughout, a linear block code $\cC$ over a finite field $\F$ is a subspace of a \emph{symbol configuration space}
$\cA = \Pi_{i=0}^{m-1} \cA_i$, where each \emph{symbol alphabet} $\cA_i$ is a finite-dimensional vector space over $\F$.

A \emph{linear tail-biting trellis realization~$\cR$ of length~$m$} is a normal linear realization, in the sense of~\cite{Fo01}
or~\cite{FGL12}, on a graph that is a single cycle of length~$m$.
Thus, it consists of a set of symbol spaces~$\cA_i$, a set of state spaces~$\cS_i$, and a set of constraint codes
$\cC_i \subseteq\cS_i \times\cA_i\times\cS_{i+1}$, where all index sets are equal to~$\Z_m$ and index arithmetic is modulo~$m$.
Every constraint code thus involves precisely two state variables.
All variable alphabets are finite-dimensional vector spaces over~$\F$.
The elements of~$\cC_i$  (called \emph{constraint codewords}, or \emph{transitions}, or \emph{branches}) will be written as
$(s_i,a_i,s_{i+1})$.

As noted in~\cite{FGL12}, if any state space $\cS_i$ is trivial, then we may simply delete it.
The graph of~$\cR$ then becomes a finite path, and the realization becomes a conventional linear trellis realization of length~$m$.
Thus finite conventional trellis realizations may be regarded as special cases of tail-biting trellis realizations.

Henceforth, we will call a linear tail-biting trellis realization simply a \emph{trellis}.

The space~$\cS=\prod_{i=0}^{m-1} \cS_i$ is called the \emph{state configuration space} of~$\cR$.
The \emph{behavior}  is the set~$\Bf$ of all \emph{trajectories} (or \emph{configurations}) $(\ab, \sb) \in \cA \times\cS$
such that all constraints are satisfied, i.e., $(s_i,a_i,s_{i+1}) \in \cC_i$ for all $i\in\Z_m$.
The pairs  $(\ab, \sb)\in\Bf$ are called \emph{valid} trajectories (or configurations).
The code $\cC$ \emph{generated by the trellis} is the set of all symbol trajectories $\ab \in \cA$ that appear in some
$(\ab, \sb) \in \Bf$.
A code generated by a linear trellis is linear.

The \emph{dual trellis of a trellis}~$\cR$, denoted by~$\cR^{\circ}$, is defined as the trellis with the same index set in which
the symbol and state spaces~$\cA_i$ and~$\cS_i$ are replaced by their linear algebra duals $\hat{\cA}_i,\,\hat{\cS}_i$
(which are unique up to isomorphism), the constraint codes~$\cC_i$ are replaced by their orthogonal codes~$\cC_i^{\perp}$
under the standard inner product, and the sign of each dual state variable is inverted in one of the two constraints in
which it is involved.
For trellises it is convenient to apply the sign inversion to~$\hat{s}_{i+1}$;
thus a dual trajectory $(\hat{\ab},\hat{\sb})$ is valid if and only if $(\hat{s}_i,\hat{a}_i,-\hat{s}_{i+1})\in\cC_i^{\perp}$
for all $i\in\Z_m$.
The behavior~$\Bf^{\circ}$ of~$\cR^{\circ}$ is the space of all such valid dual trajectories.
The \emph{Normal Realization Duality Theorem}~\cite{Fo01} states that
if~$\cR$ realizes the code~$\cC$, then its dual~$\cR^{\circ}$ realizes the orthogonal code~$\cC^\perp$.
For other proofs see~\cite{MaoKsch05,AlMao11,Fo11,FGL12}.

Trellises~$\cR$ and~$\tilde{\cR}$ of length~$m$ with state spaces and constraint codes
$\cS_i,\,\cC_i,\,\tilde{\cS}_i,\,\tilde{\cC}_i$ are called \emph{isomorphic} if there exist state-space isomorphisms
$\varphi_i:\cS_i\to\tilde{\cS}_i$
such that $\tilde{\cC}_i=\{(\varphi_i(s_i),a_i,\varphi_{i+1}(s_{i+1}))\mid (s_i,a_i,s_{i+1})\in\cC_i\}$.\footnote{It is interesting
to note that for connected trellises this definition may be relaxed to requiring that the maps~$\varphi_i$ only be
bijections rather than isomorphisms.
Indeed, Conti proved recently~\cite[Thm.~3.28, Cor.~6.6]{Con12} that the isomorphism classes of
connected trellises coincide for these two notions.}
Evidently, isomorphic trellises realize the same code.
Following~\cite[Def.~3.1]{KoVa03}, we say a trellis~$\tilde{\cR}$ is \emph{smaller} than~$\cR$ if their state spaces
satisfy $\dim\tilde{\cS}_i\leq\dim\cS_i$ for all~$i$, and we call~$\tilde{\cR}$
\emph{strictly smaller} than~$\cR$ if we have at least one strict inequality.
A trellis is called \emph{minimal} if there is no strictly smaller trellis realizing the same code.

In~\cite[Thm.~3]{FGL12} it has been shown that realizations on cycle-free graphs are minimal if and only if they are trim and
proper (see the next section for the definitions), and that a nonminimal realization can be reduced in a constructive way.
It is well known that trimness and properness are not sufficient for minimality of realizations on graphs with cycles.
The  goal of this paper is to study the particular case of single-cycle graphs,
develop constructive methods of reducing the complexity of a given realization on such a graph, and present irreducibility
criteria.
Definition~\ref{D-LocRed} below will be the central concept of this paper.

A major part of our approach will be the analysis of trellis fragments.
We use the following notation.
For $j,\, k \in\Z_m$ and $j\neq k$ let
$[j,k):=\{j,j+1,\ldots,k-1\}\subseteq\Z_m$ denote a (possibly circular) subinterval of~$\Z_m$;
thus~$\Z_m$ is the disjoint union of the two complementary subintervals $\jk$ and $[k,j)$.
Correspondingly, a trellis~$\cR$ of length~$m$ may be divided into two cycle-free \emph{complementary fragments},
denoted by $\cR^\jk$ and~$\cR^{[k,j)}$, by cutting the edges associated with states~$\cS_j$ and~$\cS_k$.
The fragment $\cR^\jk$ includes all symbol spaces and constraint codes with indices in~$\jk$,
and $\cR^{[k,j)}$ includes all with indices in $[k,j)$.
The fragment~$\cR^\jk$ also contains the state spaces~$\cS_i$ with indices in~$(j,k)$ as \emph{internal state
spaces}, and similarly~$\cR^{[k,j)}$ contains the internal state spaces~$\cS_i$ with $i\in(k,j)$.
The two boundary state spaces~$\cS_j$ and~$\cS_k$ may be regarded as \emph{external state spaces} in both fragments.
Note that the internal state spaces have degree~$2$ and correspond to normal edges, whereas the external state
spaces have degree~$1$ and correspond to half-edges, like symbol spaces.

We extend this notation by defining $[j,j+m)$ to be the entire time axis~$\Z_m$ ``starting at~$j$'', and the
complementary interval $[j+m,j)$ to be the empty interval ``starting at $j$''.
Then $\cR^{[j+m,j)}$ denotes the cycle-free fragment consisting of all of~$\cR$ except for the edge~$\cS_j$, while
$\cR^{[j+m,j)}$ denotes the complementary fragment consisting only of the edge~$\cS_j$.
Both fragments have two external state variables with common alphabet~$\cS_j$, as we will discuss further in
Section~\ref{S-Frag}.

\begin{defi}\label{D-LocRed}
Let $[j,k)$ be a non-empty interval.
A $[j,k)$-\emph{reduction}~$\tilde{\cR}$ of a trellis~$\cR$ is a
replacement of the state spaces $\cS_{j+1},\cS_{j+2},\ldots,\cS_{k-1}$ by state spaces
$\tilde{\cS}_{j+1},\tilde{\cS}_{j+2},\ldots,\tilde{\cS}_{k-1}$ of
at most the same size  and the adjacent constraint
codes $\cC_{j},\ldots,\cC_{k-1}$ by suitable constraint codes  $\tilde{\cC}_{j},\ldots,\tilde{\cC}_{k-1}$
of any size, without changing the rest of the realization or the code~$\cC$ that it realizes.
We also call this a \emph{$t$-reduction}, where $t=(k-j)\;\text{mod}\;m\in\{1,\ldots,m\}$.
The reduction will be called \emph{strict} if at least one of the state space sizes decreases
strictly, and \emph{conservative} if none of the constraint code sizes increases.
A trellis~$\cR$ is called \emph{$t$-irreducible} if each $t$-reduction is isomorphic to~$\cR$.
\end{defi}

Note that a $\jk$-reduction affects the constraint codes and the internal state spaces in the fragment $\cR^\jk$, but not
its external state spaces~$\cS_j$ and~$\cS_k$.
Thus a $t$-reduction affects~$t$ constraint codes and $t-1$ state spaces.
An $m$-reduction affects all constraint codes, and all but one state space.

If a trellis~$\cR$ is (strictly) $t$-reducible, then it is (strictly) $t'$-reducible for all $t'\geq t$.
Furthermore, it is immediate from the dualization of trellis realizations, that a trellis~$\cR$ is (strictly) $t$-reducible
if and only if the dual trellis~$\cR^{\circ}$ is (strictly) $t$-reducible.

We note in passing that a minimal trellis may have a conservative $t$-reduction for some $t>1$.
This is due to the fact that a code may have non-isomorphic minimal trellises with the same state
space and constraint code dimensions~\cite[Ex.~III.16]{GLW11}.
However, we will see later that minimal trellises are always $1$-irreducible, which is to say that no single constraint
code can be replaced by any other constraint code without changing the code realized by the trellis.

The primary goal of our reduction procedures will be the reduction of state spaces, but we will
also address the constraint code dimensions.
While state space dimensions do not change under dualization, this is not the case for the constraint
code dimensions.
This causes the constraint code dimensions to be less predictable in general.

It will become clear later that non-strict $t$-reductions form indeed a useful concept
--- even though they may not immediately lead to a net decrease of the state space sizes and
may even increase the constraint code sizes.
Non-strict reductions will be used to produce strictly reducible trellises so that ultimately a net reduction
in state space sizes is achieved.
A particular instance of a non-strict reduction is a $1$-reduction that consists of the replacement of a single
constraint code (by one that may be bigger or smaller), and thus does not alter the state complexity profile of the realization.

The most important instances of strict and conservative $2$-reductions are the mutually dual processes of trimming and merging.
We briefly recall these concepts.
Let~$\cR$ be a trellis with state spaces~$\cS_i$ and constraint codes~$\cC_i$.
Fix~$j\in\Z_m$ and let~$\cY_j$ be a subspace of~$\cS_j$.
We say that~$\cR$ is \emph{trimmed to~$\cY_j$} if we restrict the state space~$\cS_j$
to~$\cY_j$ and restrict the two adjacent constraint codes~$\cC_{j-1}$ and~$\cC_j$ accordingly.
We say that~$\cR$ is \emph{merged to the quotient space~$\cS_j/\cY_j$} if we replace the state
space~$\cS_j$ by the quotient space $\cS_j/\cY_j$ and replace the states at time~$j$ in
the two adjacent constraint codes~$\cC_{j-1}$ and~$\cC_j$ by their cosets modulo~$\cY_j$.
Projection/cross-section duality (given in~\eqref{e-PCS} below) implies that
the trellis~$\cR'$ is obtained from~$\cR$ by trimming~$\cS_j$ to
the subspace~$\cY_j$ if and only if~$(\cR')^{\circ}$ is obtained from~$\cR^{\circ}$ by merging~$\hat{\cS}_j$ to
$\hat{\cS}_j/\cY_j^{\perp}$.
For a proof, further details and a graphical illustration of the duality of trimming and merging see \cite[Sec.~III.B]{FGL12}.
In general, the trimmed/merged realization generates a different code than the original realization.
We will, of course, be interested in the case where the code does not change after trimming/merging.
In this case, trimming and merging obviously form simultaneous strict and conservative $2$-reductions of the trellis and
its dual.
For the notions of non-mergeability and non-trimmability, see the next section.

We close this section by briefly recalling the projection/cross-section duality theorem.
This identity is one of the most fundamental and useful duality relationships for linear
codes and will be used frequently throughout this paper.
Let~$\cC$ be a subspace contained in a vector space $\cT=\cT_1\times\cT_2$.
The \emph{projection} and \emph{cross-section} of~$\cC$ on $\cT_1$ are defined as
$\cC_{|\cT_1}:=\{t_1\mid \exists\; t_2\in\cT_2:\, (t_1,t_2)\in\cC\}$ and
$\cC_{:\cT_1}:=\{t_1\mid (t_1,0)\in\cC\}$, respectively.
Suppose we have inner products between~$\cT_i$ and~$\hat{\cT}_i$ for each~$i$ which we extend in the natural way to
$\cT_1\times\cT_2$ and its dual $\hat{\cT}_1\times\hat{\cT}_2$.
Projection/cross-section duality~\cite{Fo01} (see also~\cite[Sec.~II.H]{FGL12}) states that
\begin{equation}\label{e-PCS}
   \big(\cC_{:\cT_i}\big)^{\perp}=(\cC^\perp)_{|\hat{\cT}_i}
\end{equation}
for any subspace $\cC\subseteq\cT_1\times\cT_2$ and its orthogonal subspace $\cC^\perp\subseteq\hat{\cT}_1\times\hat{\cT}_2$.

\section{Local and Global Trellis Properties}\label{S-LocGlob}
We recall some basic properties of trellises as they have been discussed in detail in~\cite{FGL12}
for general normal graphs, and discuss some subtleties related to these notions.

We begin with local trellis properties.
A trellis is called \emph{trim at state space}~$\cS_i$ if $(\cC_{i-1})_{|\cS_i}=\cS_i=(\cC_i)_{|\cS_i}$,
and \emph{proper at state space}~$\cS_i$ if the cross-sections $(\cC_{i-1})_{:\cS_i}$ and $(\cC_i)_{:\cS_i}$ are both trivial.
The former means that each state in~$\cS_i$ has an incoming branch and an outgoing branch while the
latter means that there are no nontrivial branches of the form $(0,0,s_i)$ in~$\cC_{i-1}$ and none of the form
$(s_i,0,0)$ in $\cC_i$.
As in~\cite{FGL12} we call a trellis \emph{trim} (resp.\ \emph{proper}) if it is trim (resp.\ proper) at each state space.
(In the prior literature, e.g.~\cite{KoVa03} and~\cite{Va98}, ``proper'' is often called ``biproper''.)

Using projection/cross-section duality~\eqref{e-PCS} one obtains immediately the following.
\begin{theo}[\cite{GLW11a,FGL12}]\label{T-trimproper}
Let $j\in\{i-1,i\}$.
The projection of~$\cC_j$ on~$\cS_i$ is surjective if and only if the cross-section of $\cC_j^{\perp}$ on~$\hat{\cS}_i$
is trivial.
Thus,~$\cR$ is trim at~$\cS_i$ if and only if~$\cR^{\circ}$ is proper at~$\hat{\cS}_i$.
\end{theo}

Let us now turn to global trellis properties.
A trellis is said to be \emph{state-trim} if each state appears on a valid trajectory, i.e.,
$\Bf_{|\cS_i}=\cS_i$ for all~$i$.
It is clear that a state-trim trellis is trim.
We call a trellis \emph{branch-trim} if each branch appears on a valid trajectory, i.e.,
$\Bf_{|\cS_i\times\cA_i\times\cS_{i+1}}=\cC_i$ for all~$i$.
It is well known that a state-trim trellis need not be branch-trim; see, e.g., Fig.~\ref{F-BCJR}(b).
In~\cite{KoVa03}, a trellis that is both state-trim and branch-trim is called \emph{reduced}, and
all linear trellises are assumed to be reduced.
We do not adopt this stance here since the dual of a reduced trellis is not necessarily reduced (see,
for example, Fig.~\ref{F-MergTrellis}).

If a trellis or its dual is not branch-trim, then both are $1$-reducible, since in this case there exists a $1$-reduction
that is not isomorphic to the given trellis.
The converse will be discussed in Theorem~\ref{T-1Red}.

We define a trellis to be \emph{nonmergeable} (in the linear sense) if no state space can be merged to a proper quotient space without
changing the code generated by the trellis.
This is the only kind of merging that preserves linearity of the trellis.
Our notion of non-mergeability differs from the definition in~\cite{KoVa03}, which simply requires that no two states
can be merged without changing the code.
For example, the non-state-trim trellis of Fig.~\ref{F-MergTrellis}(b) below is nonmergeable in our sense, but is
mergeable in the sense of~\cite{KoVa03}, since states $00$ and $11$ in $\cS_2$ can be merged to produce a smaller
nonlinear trellis that realizes the same code.
However, it is easy to show that for state-trim trellises this situation cannot arise (since whenever
$s, s' \in \cS_i$ can be merged, then $\cS_i$ can be merged to the quotient space $\cS_i/\inner{s - s'}$ without changing the code);
see also~\cite[Observ.~2.8]{Con12}, where ``state-trim'' is called ``almost reduced''.

Dually, a trellis is called \emph{non-trimmable} if it does not allow a proper trimming (in the sense of Section~\ref{S-Basics})
resulting in a trellis for the same code.
Duality of trimming and merging, as described in Section~\ref{S-Basics}, implies that a trellis~$\cR$ is non-mergeable if
and only if its dual~$\cR^{\circ}$ is non-trimmable.
It is worth noting that a state-trim trellis may be trimmable.
This will be addressed in Theorem~\ref{P-STrimNMerg}.

A trellis~$\cR$ that realizes a code~$\cC$ is called \emph{observable} (or \emph{one-to-one}) if for each
$\ab \in \cC$ there is precisely one valid trajectory $(\ab, \sb) \in \Bf$.
The trellis is called \emph{controllable} if its dual~$\cR^{\circ}$ is observable.
These definitions are discussed in~\cite[Sec.~IV-C]{FGL12}, where it is shown that controllability in this sense is
equivalent to having independent constraints.

We will use the following controllability test from \cite[Thm.~6]{FGL12}.

\begin{theo}\label{T-ContrTest}
For every trellis~$\cR$ we have $\sum_i \dim \cC_i\leq \dim \Bf+\sum_i\dim \cS_i$, with equality if and only if~$\cR$
is controllable.
\end{theo}

In other words, Theorem~\ref{T-ContrTest} says that the total constraint dimension $\sum_i \dim \cC_i$ is maximized
if and only if the trellis is controllable.

As noted in \cite[Footnote~4]{FGL12}, this theorem is related to Theorem~4.6 of~\cite{KoVa03} as follows.
Koetter and Vardy show that if a reduced (i.e., state-trim and branch-trim) product trellis is observable
(so $\dim\Bf=\dim\cC$), and no generator has a degenerate span equal to the entire time axis~$\Z_m$, then (in our notation)
$\sum_i \dim \cC_i= \dim \Bf+\sum_i\dim \cS_i$, so the realization is controllable (in our terminology).

A key property of an uncontrollable trellis~$\cR$ is that, under weak conditions, its valid trajectories partition
into disconnected cosets (for examples, see Figs.~\ref{F-MergDual1}(b) and~\ref{F-BCJR2}(b)).
In Appendix~\ref{S-App1}, we show that this property holds provided that~$\cR$ is state-trim, which
is a weaker condition than the ``reduced'' condition of~\cite{KoVa03}.
Moreover, we give examples showing that state-trimness is indeed a necessary condition.

\section{Basic Results about Irreducibility and Motivating Examples}\label{S-Exa}
We start by presenting a list of necessary conditions for $2$-irreducibility that has been derived
in~\cite{FGL12}.
The rest of the section is devoted to examples illustrating that these properties are not sufficient and
that the dual of a trellis may be helpful in revealing the reducibility of both the trellis and its dual.

Recall that a (strict) $2$-reduction consists of reducing one state space and altering the two adjacent
constraint codes.
Obviously trimming and merging are $2$-reductions.
\begin{theo}[\mbox{\cite[Thm.~2, Thm.~9]{FGL12}}]\label{T-TPOC}
A $2$-irreducible realization~$\cR$ must be trim, proper, observable, and controllable, else
there exists a strict and conservative $2$-reduction of~$\cR$, whose dual is a strict and conservative
$2$-reduction of~$\cR^{\circ}$.
\end{theo}

In view of Theorem~\ref{T-TPOC}, we will abbreviate ``trim, proper, observable and controllable'' by~TPOC.

The following fact has been mentioned already in~\cite[Sec.~4]{Koe02}.
\begin{rem}\label{R-StateTrimNec}
Each $2$-irreducible realization must be state-trim and nonmergeable, else there exists a $2$-reduction in
form of a state-trimming or state-merging.
\end{rem}

Since $2$-reductions of unobservable trellises will be crucial later, we present the
process for further reference in the next remark.

\begin{rem}\label{R-TrimUnObs}
In the proof of \cite[Thm.~9]{FGL12} it has been shown that whenever a trellis~$\cR$ contains a nontrivial
unobservable valid trajectory $(\zerob, \sb)$, one may pick any of its nonzero states, say $s_i\in\cS_i$,
and trim the state space~$\cS_i$ to any subspace~$\cT_i$ satisfying $\cT_i\oplus\inner{s_i}=\cS_i$.
This results in a trellis realization of the same code (which may be not trim at time $i-1$ or $i+1$, in which
case it can be further reduced).
Being a trimming, this process forms a strict and conservative $2$-reduction, and the dual merging process is
a strict and conservative $2$-reduction of~$\cR^{\circ}$.
More precisely, in the trimming of~$\cR$ the branches $(s_{i-1},0,s_i)$ and
$(s_i,0,s_{i+1})$ in the constraint codes~$\cC_{i-1}$ and $\cC_i$ are deleted, and thus the dimensions of these
constraint codes decrease by~$1$.
On the other hand, the constraint code dimensions of the dual merging stay the same,
since $\dim\cC_i^{\perp}=\dim\cS_i+\dim\cS_{i+1}+\dim\cA_i-\dim\cC_i$.
\end{rem}

The following theorem is essentially due to Koetter~\cite[Thm.~9]{Koe02}.
For the second statement recall that non-mergeability is defined in the linear sense (see Section~\ref{S-LocGlob}) and
is thus dual to non-trimmability.

\begin{theo}\label{P-STrimNMerg}
A trellis is non-trimmable if and only if it is observable and state-trim.
Dually, a trellis~$\cR$ is nonmergeable if and only if the dual trellis~$\cR^{\circ}$ is
observable and state-trim.
\end{theo}
\begin{proof}
As we have just seen, an unobservable trellis is trimmable.
As noted by Koetter~\cite[Sec.~4]{Koe02}, an  observable trellis is trimmable if
and only if it is not state-trim.
The last part follows from the duality of state-merging and state-trimming discussed in Section~\ref{S-Basics}.
\end{proof}

In order to derive further necessary conditions for strict irreducibility, we present two examples of mutually dual
trellises.
Both illustrate that the dual may reveal some shortcomings of the trellis and its dual
that are not directly discernible from the primal trellis
(an observation that has been made already by Koetter in~\cite{Koe02}), and which
cause the trellis and its dual to be strictly reducible.
Motivated by these phenomena, we will study strict reducibility simultaneously for a trellis and its dual.

We assume that the reader is familiar with product trellises
\cite{KoVa02},~\cite[Sec.~IV.C]{KoVa03},~\cite[Sec.~III]{KschSo95}.
Every ``reduced'' (i.e., state-trim and branch-trim) trellis is a product trellis in the sense that
its behavior~$\Bf$ has a basis consisting of $\dim\Bf$ one-dimensional ``atomic'' sub-behaviors,
each characterized by a codeword and its (possibly circular) span.

\begin{exa}\label{E-MergDual}
We present a trellis that is TPOC, and yet strictly $2$-reducible.
This shows that Theorem~\ref{T-TPOC} provides necessary but not sufficient conditions for $2$-irreducibility.
We will see that the dual trellis, which is also TPOC, is not state-trim, and thus the primal trellis is
mergeable due to Theorem~\ref{P-STrimNMerg}.
As a consequence, both trellises are strictly $2$-reducible.
It will become clear that the non-state-trimness of the dual trellis is easy to detect, whereas the mergeability of the
primal trellis is much less obvious.

Fig.~\ref{F-MergTrellis}(a) shows a trellis realization of the
code~$\cC=\{000,110,\,101,\,011\}\subseteq\F_2^3$.\footnote{
{\bf Trellis Illustration Conventions.}
Throughout the paper, we will use the following conventions for drawing trellises:
dashed (resp.\ solid) lines denote branches with symbol variable~$0$ (resp.~$1$).
Most of the time, especially for dual pairs of trellises, a choice of state labels will be shown as well.
It should be kept in mind that the state labeling does not play any particular role and may be changed at each
index using a state space isomorphism.
The states in~$\cS_0$ will always appear in the same ordering at the beginning and end of the trellis.
}
The symbol spaces are $\cA_i=\F_2$ for $i\in\Z_3$.
The realization is the product trellis obtained from the generators $\underline{101},\,\underline{1}\,\underline{10}$ with the
indicated circular spans.
The trellis appeared first in \cite[Fig.~5]{KoVa03}, and it has been used subsequently for various purposes in
\cite[Ex.~1]{NoSh06}, \cite[Ex.~IV.8]{GLW11}, and \cite[Sec.~7.2]{Con12}; however,
its dual (Fig.~\ref{F-MergTrellis}(b)) has not been discussed in any of these papers.

\begin{figure}[ht]
\centering
    \includegraphics[height=2.5cm]{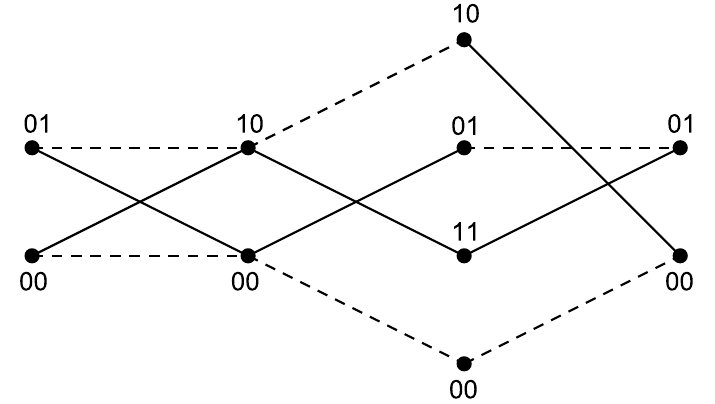}\qquad\qquad \includegraphics[height=2.5cm]{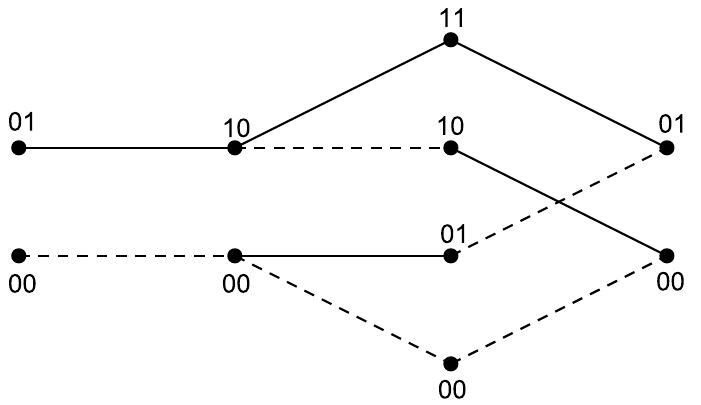}
    \\
    (a) \hspace*{4cm} (b)
    \caption{TPOC, state-trim and branch-trim trellis and its non-state-trim dual}
    \label{F-MergTrellis}
\end{figure}

\noindent Fig.~\ref{F-MergTrellis}(b) shows the dual trellis, which generates the dual code~$\cC^{\perp}=\{000,\,111\}$.
For the dualization we choose~$\hat{\cS}_i=\cS_i$ and the standard inner product for all dual state spaces; no sign
inverter is needed.
Both trellises are TPOC.
The trellis in~(a) is state-trim, whereas the dual trellis in~(b) is not state-trim:
the states~$10$ and~$01$ at time~$2$ are not on any valid trajectory of this realization.
Thus by Remark~\ref{R-StateTrimNec} both trellises are strictly $2$-reducible, and in particular,
the trellis in~(a) is mergeable due to Theorem~\ref{P-STrimNMerg}.
Trimming the dual trellis to the state space $\{00,11\}$ at time~$2$ and dually merging the primal trellis
leads to the pair of mutually dual trellises shown in Figure~\ref{F-MergDual1}.
They form $[1,0)$-reductions of the trellises in Fig.~\ref{F-MergTrellis}.

\begin{figure}[ht]
\centering
    \includegraphics[height=1.7cm]{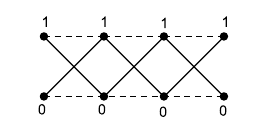}\qquad\qquad
    \includegraphics[height=1.7cm]{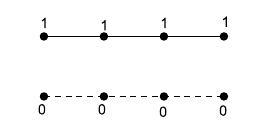}
    \\
    (a) \hspace*{5cm} (b)
    \caption{State-merged trellis and state-trimmed dual trellis}
    \label{F-MergDual1}
\end{figure}

The trellises in Fig.~\ref{F-MergDual1}(a) and~\ref{F-MergDual1}(b) are standard small examples of unobservable
and uncontrollable trellises, respectively.
As discussed in Remark~\ref{R-TrimUnObs}, both can be further reduced; the reduced trellises in this case are
conventional, trim and proper, and therefore minimal.
\QED
\end{exa}

\begin{exa}\label{E-NonMergDual}
We present a trellis that is TPOC, state-trim, branch-trim, nonmergeable, and yet strictly $2$-reducible.
The dual trellis is also TPOC, state-trim, and nonmergeable (due to the previous duality results), but
is not branch-trim.
This tells us that both trellises are reducible.
As in the previous example, the reducibility of these trellises will be obvious from the dual trellis, but is
not at all evident from the primal trellis.

The example appeared first in~\cite{GLW11,GLW11a}.
Fig.~\ref{F-BCJR}(a) shows the product trellis obtained from the generators
$0\underline{111}0,\,\underline{1}00\underline{10},\,\underline{01}\,\underline{101}$ with the
indicated  circular spans.
Its dual is shown in Fig.~\ref{F-BCJR}(b).
\begin{figure}[ht]
\centering
    \includegraphics[height=2.5cm]{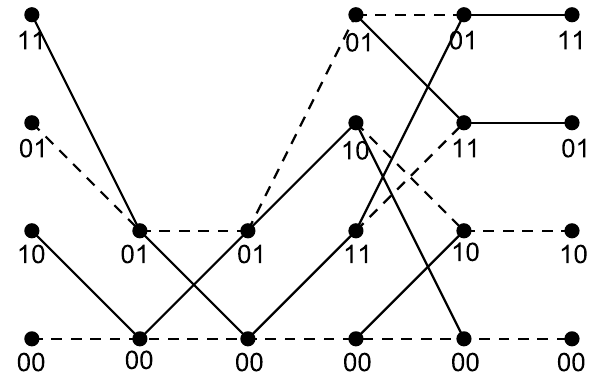}\qquad\qquad \includegraphics[height=2.5cm]{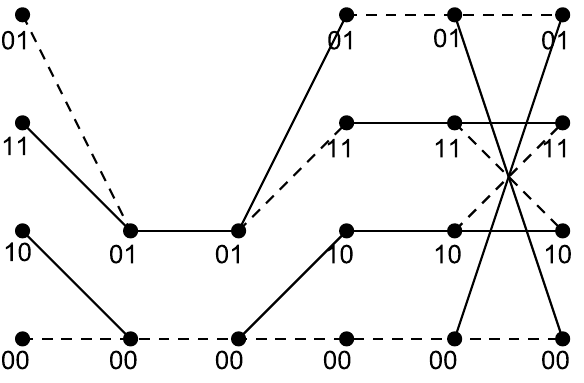}
    \\
    (a) \hspace*{4.5cm} (b)
    \caption{TPOC, state-trim, branch-trim, nonmergeable trellis and its dual}
    \label{F-BCJR}
\end{figure}

\noindent Both trellises are TPOC; they are easily seen to be state-trim, and thus both are nonmergeable due to
Theorem~\ref{P-STrimNMerg}.
The trellis in~(a) is also branch-trim, but the dual trellis is not:
the diagonal branches of the last constraint code are not on any valid trajectory, and deleting them
does not change the code generated by that trellis.
This provides us with a $1$-reduction where we replace the $3$-dimensional constraint code
$\cC_4=\inner{00|1|01,\,10|1|10,\,11|1|11}$ by the subspace $\tilde{\cC}_4=\inner{10|1|10,\,11|1|11}$.
It results in the trellis in Fig.~\ref{F-BCJR2}(b).
By duality, the trellis in Fig.~\ref{F-BCJR}(a) must also be $1$-reducible.
The dual process consists of replacing $\cC_4^{\perp}=\inner{10|0|10,\,11|1|01}$ by the
supercode $\tilde{\cC}_4^{\perp}=\inner{10|0|10,\,11|1|01,\,01|0|01}$.
This results in the trellis in Fig.~\ref{F-BCJR2}(a), which then is the dual of that in Fig.~\ref{F-BCJR2}(b).
\begin{figure}[ht]
\centering
    \includegraphics[height=2.5cm]{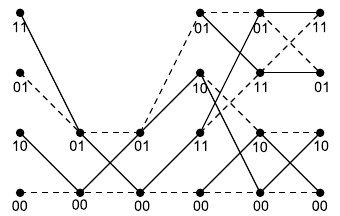}\qquad\qquad \includegraphics[height=2.5cm]{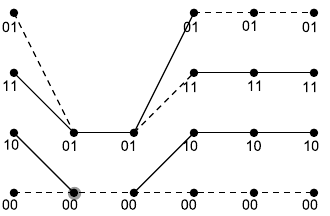}
    \\
    (a) \hspace*{5.2cm} (b)
    \caption{Branch-expanded and branch-trimmed dual trellis}
    \label{F-BCJR2}
\end{figure}

Evidently, the trellis in Fig.~\ref{F-BCJR2}(a) is unobservable and by duality (or by Theorem~\ref{T-ContrConn}) the trellis
in Fig.~\ref{F-BCJR2}(b) is uncontrollable.
Thus, we may apply the procedure from Remark~\ref{R-TrimUnObs} and trim the first trellis in a suitable way.
We pick time~$4$ and trim to the subspace $\{00,11\}$; dually, we merge the dual state space to $\F_2^2/\{00,11\}$.
This results in the mutually dual trellises shown in Fig.~\ref{F-BCJR3}.\footnote{Observe that we could also have
trimmed the state space~$\cS_4$ to~$\{00,10\}$ or the state space~$\cS_0$
suitably. Evidently these various options could lead to different results; for another one, see~\cite[Ex.~4]{FGL12}.}

\begin{figure}[ht]
\centering
    \includegraphics[height=2.5cm]{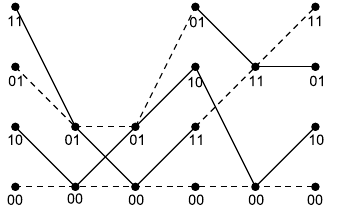}\qquad\qquad \includegraphics[height=2.5cm]{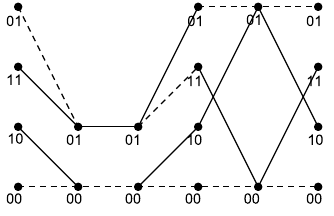}
        \\
    (a) \hspace*{5.2cm} (b)
    \caption{State-trimmed trellis and state-merged dual trellis}
    \label{F-BCJR3}
\end{figure}

\noindent These trellises still generate the original code~$\cC=\inner{01110,\,10010,\,01101}$ and its dual~$\cC^{\perp}$,
respectively.
They form strict and conservative $2$-reductions of the trellises in Fig.~\ref{F-BCJR} because no state space has changed
except~$\cS_4$, which is smaller, and only the constraint codes at times~$3$ and~$4$ have changed, and none is larger.
\QED
\end{exa}

Summarizing, we observe that the trellis in Fig.~\ref{F-BCJR}(a) is strictly $2$-reducible, even though it is
TPOC, state-trim, branch-trim, and nonmergeable.
We first had to perform an auxiliary branch-addition (a non-conservative $1$-reduction) before a state-trimming resulted in
a strict and conservative $2$-reduction.
In the next sections we will derive the appropriate concepts for a systematic study of these phenomena.

It is also worth noting that even though the trellises in Fig.~\ref{F-MergTrellis}(a) and Fig.~\ref{F-BCJR}(a)
are both product trellises, their duals are not -- simply because they are not branch-trim, a property
all product trellises share.
In Theorem~\ref{T-DualProd} we will give an intrinsic characterization of those product trellises that have a product
trellis dual.

\section{Trellis Fragments and Global Trellis Properties}\label{S-Frag}
In this section we study trellis fragments and show the duality of controllability and observability of
such fragments.
This will allow us to present some relations between various trellis properties, which then in turn leads to a
characterization of when the dual of a product trellis is a product trellis.

Let~$\cR$ be a trellis of length~$m$, and thus with symbol spaces~$\cA_i$, state spaces~$\cS_i$, and constraint
codes~$\cC_i\subseteq\cS_i\times\cA_i\times\cS_{i+1}$, all with index sets equal to~$\Z_m$.

Recall from Section~\ref{S-Basics} the definition of trellis fragments~$\cR^\jk$, where $\jk$ is any subinterval
of~$\Z_m$.
In particular, the fragment~$\cR^{[j+m,j)}$ is defined as the edge~$\cS_j$, i.e., it consists of two external
state variables, denoted by $s_{j+m}$ and~$s_j$, with a common alphabet~$\cS_j$ and an equality constraint
$\cC_= =\{(s_{j+m},s_j)\mid s_{j+m}=s_j\}$ between them.
This fragment contains no symbol spaces or internal state spaces.

The complementary fragment to~$\cR^{[j+m,j)}$ is the fragment~$\cR^{[j,j+m)}$, which consists of~$\cR$ with
the edge~$\cS_j$ cut out.
It contains all constraint codes and symbol spaces of~$\cR$ and has internal state
spaces~$\cS_i$, where $i\in(j,j+m)$, and two external state variables, with values $s_j\in\cS_j$ and
$s_{j+m}\in\cS_j$.

The \emph{internal behavior} $\Bf^\jk$ of a fragment~$\cR^\jk$ is the set of all trajectories
$(\ab^\jk,\sb^{[j,k]})\in\cA^\jk\times\cS^{[j,k]}:=\prod_{i\in[j,k)}\cA_i\times\prod_{i\in[j,k]}\cS_i$,
that satisfy all  constraints, hence $(s_i,a_i,s_{i+1})\in\cC_i$ for all $i\in[j,k)$.
Such trajectories will be called \emph{valid $\jk$-paths}, or simply \emph{valid paths}.
The \emph{external behavior} $\cC^\jk$ is the projection of $\Bf^\jk$ on $\cA^\jk\times\cS_j\times\cS_k$.
A fragment~$\cR^\jk$ is thus a normal realization of its external behavior~$\cC^\jk$.
Note that  $\Bf^{[j+m,j)}=\cC^{[j+m,j)}=\cC_=$, whereas $\Bf^{[j,j+m)}$ is the set of all valid $[j,j+m)$-paths of~$\cR$.

We will say that a valid $\jk$-path $(\ab^\jk,\sb^{[j,k]})\in\Bf^\jk$  \emph{lies on a valid trajectory} if there exists some
valid trajectory $(\ab,\sb)\in\Bf$ whose projection on $\cA^\jk\times\cS^{[j,k]}$ is $(\ab^\jk,\sb^{[j,k]})$.
In particular, a valid $[j,j+m)$-path lies on a valid trajectory if and only if $s_j=s_{j+m}$.
Notice that a traditional trellis diagram of a tail-biting trellis~$\cR$ of length~$m$, such as any trellis
diagram in this paper, actually illustrates the fragment $\cR^{[0,m)}$, and the reader has to identify~$\cS_m$ with~$\cS_0$.
The behavior $\Bf^{[0,m)}$ is not in general the same as~$\Bf$, because, again, the latter consists only of the valid
$[0,m)$-paths for which $s_0=s_m$.

All trellis fragments are cycle-free.
It follows from~\cite[Thm.~3]{FGL12} that a cycle-free fragment is a minimal realization of its external behavior
if and only if all of its constraint codes are trim and proper.

The \emph{dual fragment}~$(\cR^{\circ})^\jk$ is the dual normal realization to $\cR^\jk$, where the external state spaces
(whose degree is one) are taken as symbol spaces.
Thus $(\cR^\circ)^{[j,k)}$ comprises the dual constraint codes $(\cC_i)^\perp,  i \in [j,k)$, the dual symbol spaces
$\hat{\cA}_i, i \in [j,k)$, the internal dual state spaces $\hat{\cS}_i, i \in (j,k)$, and the external dual state spaces
$\hat{\cS}_j$ and $\hat{\cS}_k$.
In the dual realization~$(\cR^{\circ})^\jk$, we recall that a sign inversion is applied to the value of~$\hat{s}_k$ in~$\cC_{k-1}^\perp$,
but not to the value of~$\hat{s}_k$ in~$\cC_k^\perp$.
Consequently, a sign inversion is applied to one of the two occurrences  of each internal state space~$\hat{\cS}_i,\,i\in(j,k)$,
and to the one occurrence of~$\hat{\cS}_k$, but not to the one occurrence of~$\hat{\cS}_j$.

The internal behavior $(\Bf^\circ)^{[j,k)}\subseteq\hat{\cA}^\jk\times\hat{\cS}^{[j,k]}$ of~$(\cR^{\circ})^\jk$ consists of all \emph{valid}
$[j,k)$-paths $(\hat{\ab}^{[j,k)}, \hat{\sb}^{[j,k]})$, i.e., $(\hat{s}_i, \hat{a}_i, -\hat{s}_{i+1}) \in (\cC_i)^\perp$ for all
$i \in [j,k)$.
Its external behavior $(\cC^\circ)^{[j,k)}\subseteq\hat{\cS}_j\times\hat{\cA}^\jk\times\hat{\cS}_k$ is the set of all
$(\hat{s}_j,\hat{\ab}^{\jk},-\hat{s}_k)$ such that $(\hat{\ab}^{\jk},\hat{s}_j,\hat{s}_k)\in\big((\Bf^\circ)^\jk\big)_{|\hat{\cA}^\jk\times\hat{\cS}_j\times\hat{\cS}_k}$.
By normal realization duality, the external behavior of the dual fragment $(\cR^\circ)^{[j,k)}$ satisfies
$(\cC^\circ)^{[j,k)}=(\cC^{[j,k)})^\perp$.

For example, for the fragment~$\cR^{[j+m,j)}$, representing the single edge~$\cS_j$ and with behavior $\Bf^{[j+m,j)}=\cC^{[j+m,j)}=\cC_=$,
this reads as follows.
Note that the dual code to~$\cC_=$ is the \emph{sign inversion constraint code}
$\cC_{\sim}=\{(\hat{s}_{j+m},\hat{s}_j)\in\hat{\cS}_j\times\hat{\cS}_j\mid \hat{s}_{j+m}=-\hat{s}_j\}$.
Thus, the dual fragment $(\cR^\circ)^{[j+m,j)}$  has behavior
$(\Bf^{\circ})^{[j+m,j)}=\{(\hat{s}_{j+m},\hat{s}_j)\mid (\hat{s}_{j+m},-\hat{s}_j)\in\cC_{\sim}\}
=\{(\hat{s}_{j+m},\hat{s}_j)\mid \hat{s}_{j+m}=\hat{s}_j\}$, which is the equality constraint on~$\hat{\cS}_j$.
In other words, the dual fragment to an edge corresponding to~$\cS_j$ is an edge corresponding to~$\hat{\cS}_j$.
The external behavior is $(\cC^{\circ})^{[j+m,j)}=\cC_{\sim}$, the sign inversion constraint.

For any interval $[j,k)$, we define the \emph{transition space}~$\cT^\jk$ of the fragment~$\cR^\jk$ as the projection
$\mbox{$\cC^\jk$}_{|\cS_j\times\cS_k}$, and the \emph{unobservable transition space} $\cU^\jk$  as the cross-section
$\mbox{$\cC^\jk$}_{:\cS_j\times\cS_k}$, hence $\cU^\jk=\{(s_j,s_k)\in\cS_j\times\cS_k\mid (\zerob,s_j,s_k)\in\cC^\jk\}$.
Thus~$\cT^\jk$ consists of all state pairs $(s_j,s_k)$ such that there exists a valid $\jk$-path
$(\ab^\jk,s_j,\sb^{(j,k)},s_k)$, while~$\cU^\jk$ consists of all such pairs for which there is a valid $\jk$-path
with $\ab^\jk=\zerob^\jk$.
If the fragment lacks symbol spaces, as with the edge fragment~$\cR^{[j+m,j)}$, then the cross-section equals the projection,
i.e., $\cU^\jk=\cT^\jk$.

We note immediately that if~$\cC^{\jk}$ is trim and proper (i.e., $\mbox{$\cC^{\jk}$}_{|\cS_\ell}=\cS_\ell$ and
$\mbox{$\cC^{\jk}$}_{:\cS_\ell}=\{0\}$ for $\ell=j,\,k$), then $\cT^{\jk}$ is trim and $\cU^{\jk}$ is proper.

A trellis~$\cR$ will be called \emph{$\jk$-controllable} if $\cT^\jk=\cS_j\times\cS_k$, and \emph{$\jk$-observable} if
$\cU^\jk=\{(0,0)\}$.
When $[j,k)$ is a subinterval of the conventional discrete time axis~$\Z$, these definitions correspond to classical notions
of controllability and observability in linear system theory.

For a dual fragment $(\cR^\circ)^{[j,k)}$, we similarly define its transition spaces as the set of all state pairs
in $\hat{\cS}_j\times\hat{\cS}_k$ for which there exists a valid $\jk$-path, and where in addition, 
for the unobservable transition space, the symbol sequence $\ab^\jk$ is~$\zerob^\jk$.
Precisely, the \emph{transition space} is defined as
$(\cT^\circ)^{[j,k)}=\{(\hat{s}_j,\hat{s}_k)\in\hat{\cS}_j\times\hat{\cS}_k\mid \exists\;
  (\ab^{\jk},\hat{s}_j,\sb^{(j,k)},\hat{s}_k)\in(\Bf^\circ)^\jk\}=\{(\hat{s}_j,\hat{s}_k)\mid (\hat{s}_j,-\hat{s}_k)\in\big((\cC^\circ)^\jk\big)_{|\hat{\cS}_j\times\hat{\cS}_k}\}$,
and the \emph{unobservable transition space} is
$(\cU^\circ)^{[j,k)}=
  \{(\hat{s}_j,\hat{s}_k)\mid (\hat{s}_j,-\hat{s}_k)\in\big((\cC^\circ)^\jk\big)_{:\hat{\cS}_j\times\hat{\cS}_k}\}$.

A dual trellis~$\cR^\circ$ is $\jk$-\emph{controllable} if $(\cT^\circ)^{[j,k)}=\hat{\cS}_j\times\hat{\cS}_k$, and
$\jk$-\emph{observable} if $(\cU^\circ)^{[j,k)}=\{(0,0)\}$.
As a consequence, we have, just as for~$\cR$, that $\cR^\circ$ is $\jk$-controllable if
$\big((\cC^\circ)^\jk\big)_{|\hat{\cS}_j\times\hat{\cS}_k}=\hat{\cS}_j\times\hat{\cS}_k$,
and $\jk$-observable if $\big((\cC^\circ)^\jk\big)_{:\hat{\cS}_j\times\hat{\cS}_k}=\{(0,0)\}$, since the sign
inversion is evidently immaterial.
Now we obtain

\begin{theo}\label{T-FragObsContr}
If~$\cR$ and~$\cR^{\circ}$ are dual trellises, then
$\{(\hat{s}_j,\hat{s}_k)\mid (\hat{s}_j,-\hat{s}_k)\in(\cT^\circ)^\jk\}=(\cU^\jk)^{\perp}$.
In particular,~$\cR$ is $\jk$-observable if and only if its dual~$\cR^{\circ}$ is $\jk$-controllable.
\end{theo}
\begin{proof}
By $(\cC^{\circ})^\jk=(\cC^\jk)^{\perp}$ and projection/cross-section duality~\eqref{e-PCS}, we have
$\big((\cC^{\circ})^\jk\big)_{|\hat{\cS}_j\times\hat{\cS}_k}=(\mbox{$\cC^\jk$}_{:\cS_j\times\cS_k})^{\perp}
=(\cU^\jk)^{\perp}$, and the definition of $(\cT^\circ)^\jk$ yields the desired orthogonality.
The second statement follows from $\{(0,0)\}^{\perp}=\hat{\cS}_j\times\hat{\cS}_k$.
\end{proof}

We note that~$\cR$ is $[j+m,j)$-observable if and only if $\cS_j=\{0\}$.
Since~$\hat{\cS}_j=\{0\}$ if and only if~$\cS_j=\{0\}$, Theorem~\ref{T-FragObsContr} shows that~$\cR$ is also
$[j+m,j)$-controllable if and only if $\cS_j=\{0\}$.

It is worth stressing that all statements pertaining to valid $\jk$-paths, $\jk$-control\-la\-bility and
$\jk$-observability are equally valid for the primal trellis~$\cR$ and its dual~$\cR^\circ$.
The only slight asymmetry, due to sign inversion, is contained in the external behavior (and thus in the transition spaces),
and  has been dealt with in the previous result.
From this point on, no distinction needs to be made between a primal and a dual trellis.

We next discuss  various global notions of trimness.
A trellis~$\cR$ will be called \emph{${[j,k)}$-trim} if every valid $[j,k)$-path $(\ab^{[j,k)}, \sb^{[j,k]})$ lies
on a valid trajectory $(\ab,\sb)\in\Bf$.
Evidently~$\cR$ is $[j,k)$-trim if and only if all state pairs $(s_j, s_k)$ in the transition space~$\cT^{[j,k)}$
also occur in~$\cT^{[k,j)}$;  i.e., $\cT^{[j,k)} \subseteq \cT^{[k,j)}$.
Thus if~$\cR$ is $[k,j)$-controllable, so $\cT^{[k,j)} = \cS_j \times \cS_k$, then~$\cR$ is $[j,k)$-trim.
However, the converse is not true unless we require also that~$\cR$ be controllable.

\begin{theo}\label{T-ContrTrim}
(a) A $[k,j)$-controllable trellis is $[j,k)$-trim.
\\
(b) A controllable and $[j,k)$-trim trellis is $[k,j)$-controllable.
\end{theo}
\begin{proof}
We have already shown (a).
To prove~(b), we use the dual trellis~$\cR^\circ$.
As we have seen above, $[j,k)$-trimness of~$\cR$ implies $\cT^{[j,k)} \subseteq \cT^{[k,j)}$.
By Theorem~\ref{T-FragObsContr} this yields $(\cU^\circ)^{[k,j)} \subseteq (\cU^\circ)^{[j,k)}$.
But this means that for every $(\hat{s}_j, \hat{s}_k) \in (\cU^\circ)^{[k,j)}$ there is a valid
trajectory in~$\Bf^\circ$ with $\hat{\ab} = \zerob$.
Since~$\cR^\circ$ is observable,~$(\cU^\circ)^{[k,j)}$ must be trivial; i.e.,~$\cR^\circ$ must be $[k,j)$-observable.
By Theorem~\ref{T-FragObsContr},~$\cR$ must be $[k,j)$-controllable.
\end{proof}

We now introduce notions of controller and observer memory similar to those of classical linear system theory.
For $1 \le t \le m$, we will say that~$\cR$ is \emph{$t$-controllable} (resp.\ \emph{$t$-observable}) if~$\cR$ is
$[j,j+t)$-controllable (resp.\ $[j,j+t)$-observable) for all length-$t$ intervals $[j,j+t)$.

In particular, a trellis~$\cR$ is $m$-controllable if and only if for all~$j$ there is a valid $[j,j+m)$-path of length~$m$
from any state $s_j \in \cS_j$ to any state $s_{j+m} \in \cS_j$.
Thus $m$-controllable trellises are not only trim, but also state-trim.
Dually, $m$-observable trellises are proper.
Moreover, $m$-observable trellises are evidently observable, which implies that $m$-controllable trellises are controllable.
However, the converses of these statements are not necessarily true, as we now proceed to show.

We first note that a trellis is state-trim at $\cS_j$ if and only if every valid
$[j+m,j)$-path $(s_{j+m}, s_j) \in \cC^{[j+m, j)}$ lies on a valid trajectory in $\Bf$.
Then we obtain the following corollary:

\begin{cor}\label{C-ContrTrim}
A controllable trellis is $m$-controllable if and only if it is state-trim.
Dually, an observable trellis~$\cR$ is $m$-observable if and only if its dual~$\cR^\circ$ is state-trim.
Consequently, a nonmergeable observable trellis is $m$-observable.
\end{cor}
\begin{proof}
By Theorem~\ref{T-ContrTrim}, a controllable trellis~$\cR$ is $[j,j+m)$-controllable for all~$j$ if and only if it is
$[j+m,j)$-trim for all~$j$, which is to say if it is state-trim at~$\cS_j$ for all~$j$.
The second statement follows from Theorem~\ref{T-FragObsContr}.
The last statement follows from Theorem~\ref{P-STrimNMerg}, which shows that if~$\cR$ is non-mergeable, 
then $\cR^\circ$ must be state-trim.
\end{proof}

For example, Fig.~\ref{F-MergTrellis}(a) shows an observable trellis that is not $m$-observable (note the all-zero path
between states $s_2 = 01$ and $s_2' = 10$);
its dual in Fig.~\ref{F-MergTrellis}(b) is a controllable trellis that is not $m$-controllable, and not state-trim at~$\cS_2$.

The last observation of Corollary~\ref{C-ContrTrim} also appears  in \cite[Thm.~7.8]{Con12}, where an $m$-observable
trellis is called ``totally one-to-one."

Similarly, we note that a trellis is branch-trim at constraint code~$\cC_i$ if and only if every valid
$[i,i+1)$-path $(s_i, a_i, s_{i+1}) \in \cC_i$ lies on a valid trajectory in~$\Bf$.
Thus we obtain the following corollary:

\begin{cor}\label{C-ContrBranchTrim}
A controllable trellis is $(m\!-\!1)$-controllable if and only if it is branch-trim.
Dually, an observable trellis~$\cR$ is $(m\!-\!1)$-observable if and only if its dual~$\cR^\circ$ is branch-trim.
\end{cor}
\begin{proof}  By Theorem~\ref{T-ContrTrim}, a  controllable trellis~$\cR$ is $[i+1,i)$-controllable for all~$i$
if and only if it is $[i,i+1)$-trim for all~$i$, which is to say if it is branch-trim at~$\cC_i$ for all $i$.
The second statement follows from Theorem~\ref{T-FragObsContr}.
\end{proof}

For example, Fig.~\ref{F-BCJR}(a) shows an observable trellis that is not $(m\!-\!1)$-observable (note the all-zero
path between states $s_5 = 01$ and $s_4 = 01$);  its dual in Fig.~\ref{F-BCJR}(b) is a controllable trellis
that is not $(m\!-\!1)$-controllable, and not branch-trim at~$\cC_4$.

We remark that in the first statement of Corollary 5.4, controllability is necessary;  for example, Figs.~\ref{F-MergDual1}(b)
and~\ref{F-BCJR2}(b) show uncontrollable trellises that are branch-trim, but not $t$-controllable for any $t \le m$.

Combining these corollaries, we obtain the following results.
Part~(b) can also be found in~\cite[Thm.~7.8]{Con12}.

\begin{cor}\label{C-ContrTrim2}
(a) An $(m\!-\!1)$-controllable trellis is branch-trim, state-trim and trim.
\\
(b)  A controllable $m$-observable trellis is nonmergeable and proper.
\end{cor}
\begin{proof}
(a) If~$\cR$ is $(m\!-\!1)$-controllable, then it is also $m$-controllable, so it is not only branch-trim by
Corollary~\ref{C-ContrBranchTrim}, but also state-trim by Corollary~\ref{C-ContrTrim}, and thus trim.
(b)  If~$\cR$ is $m$-observable, then its dual~$\cR^\circ$ is $m$-controllable, and thus state-trim by
Corollary~\ref{C-ContrTrim}.
Theorem~\ref{P-STrimNMerg} implies that~$\cR$ is nonmergeable, which in turn implies properness.
\end{proof}

Controllability is necessary in~(b) because an uncontrollable trellis is mergeable;  for example,
the trellis of Fig.~\ref{F-MergDual1}(b) is $3$-observable but uncontrollable, hence mergeable.

Finally, we address the question of when the dual of a product trellis is a product trellis.
Koetter and Vardy [10] showed that a trellis is a product trellis if and only if it is ``reduced"
(state-trim and branch-trim), which gave them a powerful tool in their search for minimal trellises~\cite{KoVa03}.
But Examples~\ref{E-MergDual} and~\ref{E-NonMergDual} show that the dual of a product trellis is not
necessarily reduced.
For observable product trellises, Corollaries~\ref{C-ContrBranchTrim} and~\ref{C-ContrTrim2}
give us a nice characterization of when the dual is a product trellis.

\begin{theo}\label{T-DualProd}
(a) If~$\cR$ is an $(m\!-\!1)$-observable trellis, then its dual~$\cR^\circ$ is a product trellis.
\\
(b) If~$\cR$ is observable but not $(m\!-\!1)$-observable, then its dual~$\cR^\circ$ is not a product trellis.
\end{theo}
\begin{proof}  (a) If~$\cR$ is $(m\!-\!1)$-observable, then~$\cR^\circ$ is $(m\!-\!1)$-controllable, thus reduced by
Corollary~\ref{C-ContrTrim2}, thus a product trellis.
(b)  If~$\cR$ is observable but not $(m\!-\!1)$-observable, then~$\cR^\circ$ is controllable but not $(m\!-\!1)$-controllable,
thus not branch-trim by Corollary~\ref{C-ContrBranchTrim}, thus not a product trellis.
\end{proof}

For example, the observable trellis of Fig.~\ref{F-MergTrellis}(a) is not $t$-observable for any $t \le m$;
its dual in Fig.~\ref{F-MergTrellis}(b) is neither state-trim nor branch-trim.
For another example, the observable trellis of Fig.~\ref{F-BCJR}(a) is $m$-observable but not $(m\!-\!1)$-observable;
its dual in Fig.~\ref{F-BCJR}(b) is not branch-trim.

We remark that Theorem~\ref{T-DualProd}(b) may be extended to unobservable proper trellises as follows.
(By Theorem~\ref{T-trimproper}, the dual of an improper trellis is not trim, hence not reduced.)
We have to redefine $[j,k)$-observability as follows.
Given a trellis~$\cR$ with behavior~$\Bf$, the \emph{unobservable state configuration space} is defined as
$\cS^{\mathrm{u}} = \Bf_{:\cS} = \{\sb \in \cS : (\zerob, \sb) \in \Bf\}$~\cite{FGL12}.
Then~$\cR$ is called $[j,k)$-observable if the unobservable transition space~$\cU^{[j,k)}$ equals the projection
$(\cS^{\mathrm{u}})_{|\cS_j \times \cS_k}$; i.e., if~$\cR$ does not contain any valid $[j,k)$-paths with
$\ab^{[j,k)} = \zerob^{[j,k)}$ other than those that lie on valid unobservable trajectories
$(\zerob, \sb) \in \Bf$.
Then we can show that if~$\cR$ is proper but not $(m\!-\!1)$-observable in this sense, then~$\cR^\circ$ is not a product trellis.

\section{Constructing Reductions}\label{S-LocIrr}
This section begins to establish the main results of this paper concerning whether a trellis~$\cR$ is $[k,j)$-reducible.
First, we shall give sufficient conditions for the $[k,j)$-irreducibility of~$\cR$, i.e., for~$\cR$ not having any
$[k,j)$-reduction other than itself, up to isomorphism.
Second, when these conditions are not met, but another auxiliary condition is met, we will construct a strict and
conservative $[k,j)$-reduction.

Due to Theorem~\ref{T-TPOC} we may assume without loss of generality that~$\cR$ is TPOC.

\begin{theo}\label{T-jkObsIrr}
Let $j\neq k{}\!\mod{}m$. Suppose that both~$\cR$ and~$\cR^{\circ}$ are TPOC and $[j,k)$-observable.
Then both~$\cR$ and~$\cR^{\circ}$ are $[k,j)$-irreducible.
\end{theo}
\begin{proof}
It suffices to consider~$\cR$.
Let $\tilde{\cR}$ be a $[k,j)$-reduction of~$\cR$.
Without loss of generality we may assume that~$\tilde{\cR}$ is trim and proper at time $k+1,\ldots,j-1$ since otherwise
we may reduce further.
We must show that~$\cR$ and~$\tilde{\cR}$ are isomorphic.

Since $\tilde{\cR}^{[j,k)} = \cR^{[j,k)}$, the trellis~$\tilde{\cR}$ is $\jk$-observable and thus observable.
Furthermore,~$\tilde{\cR}$ is $\jk$-controllable, and thus $[k,j)$-trim by Theorem~\ref{T-ContrTrim}.
Using trimness and $[k,j)$-trimness of~$\cR$, we conclude that~$\cR$ is state-trim at times~$k$ and~$j$.
But then $\tilde{\cR}^{[j,k)} = \cR^{[j,k)}$ and $\jk$-observability imply that~$\tilde{\cR}$
is also state-trim, thus trim, at times~$k$ and~$j$.

Thus it remains to show that the fragments~$\cR^{[k,j)}$ and~$\tilde{\cR}^{[k,j)}$ are isomorphic.
In order to do so we show first $\cC^{[k,j)}=\tilde{\cC}^{[k,j)}$.
Let $(s_k,\ab^{[k,j)}, s_j)\in\cC^{[k,j)}$.
Then there is a path $(\ab^{[k,j)},\sb^{[k,j)})\in\Bf^{[k,j)}$, and
by $[k,j)$-trimness this path lies on a trajectory in~$\cR$, say $(\ab, \sb) \in \Bf$.
Observability of~$\tilde{\cR}$ implies that there is a unique trajectory $(\ab, \tilde{\sb})$ in 
the behavior~$\tilde{\Bf}$
of~$\tilde{\cR}$, and $\tilde{\cR}^{[j,k)} = \cR^{[j,k)}$ along with
$\jk$-observability yields $(\ab^{[j,k)}, \sb^{[j,k]})=(\ab^{[j,k)}, \tilde{\sb}^{[j,k]})$.
Hence $(s_j,s_k)=(\tilde{s}_j,\tilde{s}_k)$, and this proves that $(s_k,\ab^{[k,j)}, s_j)\in\tilde{\cC}^{[k,j)}$.
In the same way one concludes that $\tilde{\cC}^{[k,j)}\subseteq\cC^{[k,j)}$.

All of this shows that~$\cR^{[k,j)}$ and~$\tilde{\cR}^{[k,j)}$ are both
trim and proper cycle-free trellis fragments that realize the same external behavior.
By~\cite[Thm.~3]{FGL12} they are both minimal, and must be isomorphic.
But then~$\cR$ and~$\tilde{\cR}$ are isomorphic, and this concludes the proof.
\end{proof}

We note in passing that Theorem~\ref{T-jkObsIrr} is also true in the case where
$j=k{}\!\mod{}m$, in which it reproduces earlier results.
On the one hand, if~$\cR$ and~$\cR^\circ$ are $[j,j+m)$-observable, hence $[j,j+m)$-controllable,
then the proof of Corollary~\ref{C-ContrTrim} implies that they are both state-trim at time~$j$,
and this may be regarded as $[j+m,j)$-irreducibility (recall the equality constraint~$\cC_=$).
On the other hand, we have seen already that~$\cR$ and~$\cR^{\circ}$ are $[j+m,j)$-observable if and only
if~$\cS_j=\{0\}$;
in this case they are conventional TPOC trellises, hence minimal, and thus irreducible on any interval.

\medskip

Next, we recall that, as we have already seen in Example~\ref{E-NonMergDual} and Corollary~\ref{C-ContrBranchTrim}, a
controllable trellis that is not $(m\!-\!1)$-controllable is not branch-trim, and thus has a conservative $1$-reduction
consisting of branch-trimming some constraint code, as in Fig.~\ref{F-BCJR2}(b).
Dually, an observable trellis that is not $(m\!-\!1)$-observable has a non-conservative $1$-reduction consisting of
branch-expanding some constraint code, as in Fig.~\ref{F-BCJR2}(a).
But this ``reduced" trellis must be unobservable, and hence may be state-trimmed so as to make the reduction both strict and
conservative, as in Fig.~\ref{F-BCJR3}(a).
We record these observations as a lemma:

\begin{lemma}\label{L-m1UnObs}
If a trellis $\cR$ of length~$m$ is observable but not $(m\!-\!1)$-observable, then it has a strict and conservative $2$-reduction.
Its dual is a strict and conservative $2$-reduction of~$\cR^{\circ}$.
\end{lemma}
\begin{proof}
The dual trellis $\cR^\circ$ is controllable but not $(m\!-\!1)$-controllable, and thus by Corollary~\ref{C-ContrBranchTrim}
is not branch-trim at some constraint code~$(\cC_i)^\perp$.
Replacing $(\cC_i)^\perp$ by a suitably branch-trimmed~$(\tilde{\cC}_i)^\perp$ reduces the dimension of~$(\cC_i)^\perp$ by one,
without changing the realized code~$\cC^\perp$.
Dually, replacing~$\cC_i$ by~$\tilde{\cC}_i$ expands~$\cC_i$ by one dimension, without changing the realized code~$\cC$.
Denote the resulting realizations by~$\tilde{\cR}^\circ$ and~$\tilde{\cR}$, respectively.

By Theorem~\ref{T-ContrTest}, $\tilde{\cR}^\circ$ must be uncontrollable, since we have reduced the dimension of a
constraint code without changing~$\Bf$ or~$\cS$.
Hence~$\tilde{\cR}$ is unobservable.
As shown in Remark~\ref{R-TrimUnObs}, we can trim~$\tilde{\cR}$ at any state space without changing the code~$\cC$.
Trimming the state space~$\cS_i$  reduces the dimensions of~$\cC_{i-1},\, \cS_i,$ and $\tilde{\cC}_i$ by one,
thus achieving a strict and conservative $2$-reduction.
By construction, the dual reduction is also strict and conservative.
\end{proof}

We now generalize Lemma~\ref{L-m1UnObs} to trellises that are not $(m\!-\!t)$-observable, provided that the trellis satisfies a
certain technical condition.
Without loss of generality, we may assume that the trellis is not $[0,m\!-\!t)$-observable.

\begin{theo}\label{T-ZeroRun}
Let $2 \leq t \leq m - 1$.
Let~$\cR$ be a TPOC trellis of length $m$ that is not $[0,m-t)$-observable;
i.e.,~$\cU^{[0,m-t)}$ is nontrivial, so (by properness) there is some $s_0 \neq 0,\, s_{m-t} \neq 0$ such that there is an
unobservable path from~$s_0$ to~$s_{m-t}$ in $\cR^{[0,m-t)}$.
Suppose~$\cR$ satisfies one of the following two conditions:
\\[.5ex]
\underline{Condition A:} In the fragment~$\cR^{[m-t,m-1)}$, there is no valid path from~$s_{m-t}$ to~$0 \in \cS_{m-1}$.
\\[.5ex]
\underline{Condition A$'$:} In the fragment~$\cR^{[m-t+1,0)}$, there is no valid path from $0 \in \cS_{m-t+1}$ to~$s_0$.
\\[.5ex]
Then~$\cR$ has a conservative $t$-reduction and a strict and conservative $(t+1)$-reduction.
For each of these reductions the dual is a reduction of the same type of the dual trellis~$\cR^{\circ}$.
\end{theo}

The $t$-reduction stated in the theorem is the main step of the reduction process and leads immediately to the
strict $(t+1)$-reduction.
Both parts will be used in Theorem~\ref{T-Charac-chi}.

We give a sketch of the proof and outline the reduction procedure.
The details are carried out in Appendix~\ref{S-AppA}.

\noindent\textsc{Sketch of Proof}:
By assumption there exists an unobservable valid path from~$s_0 \neq 0$ to~$s_{m-t} \neq 0$ in $\cR^{[0,m-t)}$.
However, since $\cR$ is observable, there can be no unobservable valid path from~$s_{m-t}$ to~$s_0$ in the
complementary fragment $\cR^{[m-t,0)}$.

\noindent\underline{Step~1:}
We expand $\cR^{[m-t,0)}$ so that it contains an unobservable path $(\zerob^{[m-t,0)}, \tilde{\sb}^{[m-t,0]})$
from~$\tilde{s}_{m-t}:=s_{m-t} $ to~$\tilde{s}_0:=s_0$ via a sequence of new states $\tilde{s}_i \notin \cS_i$ for $i \in (m-t,0)$.
The internal state spaces and the constraint codes of $\cR^{[m-t,0)}$ are expanded to
\[
  \cS^+_i = \cS_i \oplus \langle\tilde{s}_i\rangle\text{ for } i \in (m-t,0)\ \text{ and }\
  \cC^+_i = \cC_i \oplus \langle (\tilde{s}_i, 0, \tilde{s}_{i+1})\rangle\text{ for } i \in [m-t,0).
\]
It is straightforward to show that the internal behavior of the fragment~$\cR^{[m-t,0)}$
consequently expands to $(\Bf^{[m-t,0)})^+ = \Bf^{[m-t,0)} \oplus \langle (\zerob^{[m-t,0)}, \tilde{\sb}^{[m-t,0]}) \rangle$.

By construction, the expanded trellis $\cR^+$ is unobservable and has a valid trajectory $(\zerob, \tilde{\sb})$ that
passes through~$s_0$ and~$s_{m-t}$.
It is straightforward to show that its behavior is $\Bf^+ = \Bf \oplus \langle(\zerob, \tilde{\sb})\rangle$,
so it continues to realize the same code.

\noindent\underline{Step~2:}
Assume that~$\cR$ satisfies Condition~A. Then this guarantees (see Appendix~\ref{S-AppA}) that we can find a strict subspace
$\tilde{\cS}_{m-1} \subset \cS^+_{m-1}$ such that there is no valid path from~$s_{m-t}$ to any state in~$\tilde{\cS}_{m-1}$.
We trim $\cS^+_{m-1}$ to $\tilde{\cS}_{m-1}$.

\noindent\underline{Step~3:}
By Condition~A, the resulting trellis will not be trim at times $m-2,\,m-3,\ldots, m-t$, and so one can successively trim the subspaces
$\cS^+_{m-2}, \ldots,\, \cS^+_{m-t+1},\,\cS_{m-t}$ by at least one dimension.
This will also reduce the adjacent constraint codes $\cC_{m-1}^+,\,\cC_{m-2}^+,\ldots,\,\cC_{m-t}^+,\,\cC_{m-t-1}$ by at least one dimension.
All of this results in a strict and conservative $[m-t-1,0)$-reduction whose dual is of the same form.
Trimming only the state spaces $\cS^+_{m-2}, \ldots, \cS^+_{m-t+1}$ and the adjacent constraint codes
$\cC_{m-1}^+,\,\cC_{m-2}^+,\ldots,\,\cC_{m-t}^+$ results in the stated conservative $t$-reduction.
\hfill$\Box$

\medskip

We illustrate this procedure with an example.

\noindent{\bf Example~\ref{E-NonMergDual}$'$}
The trellis of Fig.~\ref{F-BCJR3}(a) is TPOC but not $[0,3)$-observable.
It satisfies Condition~A, in that there is no transition in~$\cC_4$ from $01 \in \cS_3$ to
$00 \in \cS_4$.
Figure~\ref{F-BCJRRed} illustrates the initial expansion of~$\cS_4$ to $\cS_4^+ = \langle 1, \tilde{s} \rangle$ and the
corresponding expansion of the adjacent constraint codes,~$\cC_3$ and~$\cC_4$.

\pagebreak

\begin{figure}[ht]
\centering
    \includegraphics[height=3cm]{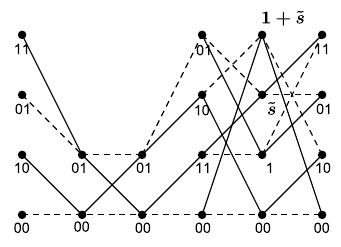}
    \caption{Expanding the trellis in Fig.~\ref{F-BCJR3}(a)}
    \label{F-BCJRRed}
\end{figure}

\noindent Next, we identify $\tilde{\cS}_{4} = \langle 1 + \tilde{s} \rangle$ as a strict subspace of~$\cS_4^+$ such that there is no
valid path from $01 \in \cS_3$ to $\tilde{\cS}_{4}$.
Then $\cS_4^+$ is trimmed to $\tilde{\cS}_{4}$, as in Fig.~\ref{F-BCJRRed2}, which also reduces $\cC_3^+$ and $\cC_4^+$.
\begin{figure}[ht]
\centering
    \includegraphics[height=3cm]{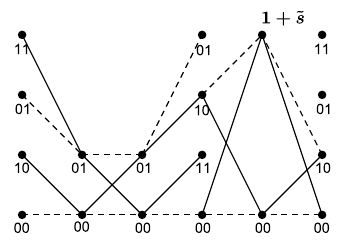}
    \\
    \caption{Trimming the trellis in Fig.~\ref{F-BCJRRed}}
    \label{F-BCJRRed2}
\end{figure}

\noindent The resulting trellis is guaranteed not to be trim at~$\cS_3$, so~$\cS_3$ and~$\cC_2$ may be trimmed to achieve a
strict and conservative 3-reduction.
Moreover, in this example, if we trim all states and branches that are not on any valid trajectories, then we reach the
conventional (hence minimal) trellis shown in Fig.~\ref{F-BCJRRed3}.
\begin{figure}[ht]
\centering
    \includegraphics[height=1.7cm]{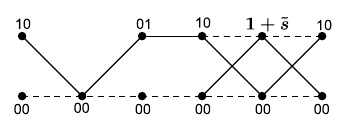}
    \\
    \caption{Trimming the trellis in Fig.~\ref{F-BCJRRed2}}
    \label{F-BCJRRed3}
\end{figure}

\vspace*{-5.5ex}
\QED

\section{$t$-Irreducibility}\label{S-tRed}
Now we are in a position to discuss $t$-irreducibility for all values of~$t$ not larger than a certain invariant of the code.
We will see that it is characterized by $(m-t)$-observability and $(m-t)$-controllability.
Thereafter we will illustrate the open problems that arise for larger values of~$t$.
Due to Theorem~\ref{T-TPOC} we may restrict attention to TPOC trellises.

We begin with $1$-reductions, that is, the replacement of one constraint code.
The following intrinsic characterization of $1$-irreducibility shows in particular that every $1$-reduction
is either a branch-trimming  (replacement of a constraint code by a strict subcode) or a branch-addition
(replacement of a constraint code by a strict supercode).

\pagebreak

\begin{theo}\label{T-1Red}
Let~$\cR$ be TPOC. Then the following are equivalent.
\begin{romanlist}
\item $\cR$ and~$\cR^{\circ}$ are $1$-irreducible,
\item $\cR$ and~$\cR^{\circ}$ are branch-trim,
\item $\cR$ and~$\cR^{\circ}$ are $(m\!-\!1)$-observable and $(m\!-\!1)$-controllable.
\end{romanlist}
Furthermore, if~$\cR$ is $1$-reducible, then it has a strict and conservative $2$-reduction~$\tilde{\cR}$
whose dual~$\tilde{\cR}^{\circ}$ is a strict and conservative $2$-reduction of~$\cR^\circ$.
\end{theo}
\begin{proof}
(i)~$\Rightarrow$~(ii) is obvious.
(ii)~$\Rightarrow$~(iii) follows from Corollary~\ref{C-ContrBranchTrim}, and
(iii)~$\Rightarrow$~(i) is a special case of Theorem~\ref{T-jkObsIrr}.
For the last statement, let $\cR$ be $1$-reducible.
Then either the trellis or its dual is $(m\!-\!1)$-unobservable and the result follows from
Lemma~\ref{L-m1UnObs}.
\end{proof}

An example for the reduction of an $(m\!-\!1)$-unobservable trellis has been shown already in Example~\ref{E-NonMergDual}.

Recall that minimality of a trellis is defined on the basis of its state space dimensions.
From the above we obtain the following corollary, which tells us that the constraint codes of a minimal trellis
cannot be replaced by smaller (or larger) constraint codes.
In particular, minimal trellises are branch-trim, which has also recently been established by different arguments
in~\cite[Thm.~7.1]{Con12}.
One may recall that in~\cite{KoVa03} Koetter/Vardy restrict themselves to state-trim and branch-trim
trellises, so that in their terminology minimal trellises are branch-trim by definition.

\begin{cor}\label{C-TMin}
A minimal trellis is $1$-irreducible.
\end{cor}

We now consider $t$-irreducibility for $t>1$.
We first introduce the following parameter of a code~$\cC$, which derives from the ``characteristic spans'' of
Koetter/Vardy~\cite{KoVa03} (see also~\cite{GLW11,GLW11a}).

\begin{defi}\label{D-span}
The \emph{minimum span length}~$\chi(\cC)$ of a code $\cC \subseteq \cA = \prod_i \cA_i$ is the minimum length of all possible spans
of the nonzero codewords~$\ab \in \cC$, where a span of $\ab \neq \zerob$ is any interval,
possibly circular, that covers the support of $\ab$.
\end{defi}

For example, the code $\cC=\{00000,\,10110,\,11001,\,01111\}\subset\F_2^5$  has  minimum span
length \mbox{$\chi(\cC)=3$}.

Now we can formulate the following characterization of $t$-irreducibility.
The proof provides us with a constructive reduction method for $t$-reducible trellises.

\begin{theo}\label{T-Charac-chi}
Let $\min\{\chi(\cC),\,\chi(\cC^{\perp})\}>t>1$ and let~$\cR$ be a TPOC trellis of~$\cC$.
Then the following are equivalent.
\begin{romanlist}
\item $\cR$ (and thus~$\cR^{\circ}$) is $(m-t)$-observable and $(m-t)$-controllable.
\item $\cR$ (and thus~$\cR^{\circ}$) is $t$-irreducible.
\end{romanlist}
For a $t$-reducible trellis~$\cR$ we have the following cases:
\\
if~$\cR$ is $(m-t+1)$-unobservable or $(m-t+1)$-uncontrollable then~$\cR$ allows a strict and conservative $t$-reduction.
If~$\cR$ is $(m-t+1)$-observable and $(m-t+1)$-controllable, then~$\cR$ allows non-strict and conservative $t$-reduction which
gives rise to a subsequent strict and conservative $(t+1)$-reduction.
In either case, the dual process is a reduction of the same type for~$\cR^{\circ}$.
\end{theo}
\begin{proof}
(i)~$\Rightarrow$~(ii) is Theorem~\ref{T-jkObsIrr}.
As for the converse assume without loss of generality that~$\cR^{[0,m-t)}$ is unobservable, and let~$s_0$ and~$s_{m-t}$ be the
end states of a nontrivial unobservable valid path of this fragment.
Since~$\cR$ is $t$-irreducible, Theorem~\ref{T-ZeroRun} tells us that there must exist valid paths
from~$s_{m-t}$ to $0\in\cS_{m-1}$ in the fragment $\cR^{[m-t,m-1)}$ and from $0\in\cS_{m-t+1}$ to~$s_0$ in $\cR^{[m-t+1,0)}$.
Adding those two paths (suitably appended by zero branches) results in a valid path from $s_{m-t}$ to $s_0$.
But this means that the given unobservable path lies on a valid trajectory in~$\cR$.
Hence it represents a codeword in~$\cC$ that has a span of length at most~$t$, and this
contradicts our assumption.

From Theorem~\ref{T-ZeroRun} it is clear that every $t$-reducible trellis allows a non-strict conservative
$t$-reduction that is non-trim and thus gives rise to a subsequent strict and conservative $(t+1)$-reduction.
Suppose now that~$\cR$ is $(m-t+1)$-unobservable.
Hence there exists an unobservable valid path of length~$m-t+1$, say in $\cR^{[0,m-t+1)}$.
With the aid of $\chi(\cC)>t$, we conclude that this path cannot lie on a valid trajectory, which in turn
means that Condition~A or~A$'$ of Theorem~\ref{T-ZeroRun} must be satisfied.
Consequently, by that theorem there exists a conservative $(t-1)$-reduction
and a strict and conservative $t$-reduction, whose duals are reductions of the same type.
\end{proof}

We remark that Theorem~\ref{T-Charac-chi} applies to conventional trellises as follows.
A trim and proper conventional trellis~$\cR$ and its dual~$\cR^\circ$ are minimal and therefore $t$-irreducible
for all~$t$.
Hence by Theorem~\ref{T-Charac-chi}, $\cR$ and $\cR^\circ$ must be $(m-t)$-observable for all
$t < \min\{\chi(\cC), \chi(\cC^\perp)\}$.
But this is obvious, since in a minimal conventional trellis an unobservable path of length~$m-t$ would imply a nonzero
codeword with circular span length~$t$ or less, since all valid paths lie on valid trajectories.
In this sense, the limit $t < \min\{\chi(\cC), \chi(\cC^\perp)\}$ in Theorem~\ref{T-Charac-chi} is the best possible.

Let us now return to the general situation.
Theorem~\ref{T-Charac-chi} characterizes $t$-irreducibility for small values of~$t$.
For $t\geq\min\{\chi(\cC),\,\chi(\cC^{\perp})\}$ the implication (i)~$\Rightarrow$~(ii)
remains valid (being solely based on Theorem~\ref{T-jkObsIrr}), whereas the converse is not true in general.
For instance, the trellis in Fig.~\ref{F-BCJRRed2} is conventional and minimal, thus $t$-irreducible
for each~$t$, but it is not $2$-observable.

The next example illustrates that in some cases, $t$-reducibility for $t\geq\min\{\chi(\cC),\,\chi(\cC^{\perp})\}$
may follow directly from Theorem~\ref{T-ZeroRun}.

\begin{exa}\label{E-2Irred}
The trellis~$\cR$ in Fig.~\ref{F-2O2C} is the product trellis obtained from the generators
$\underline{11011}0000$, $0\underline{101}00000,\,0000\underline{1101}0,\,\underline{1}000000\underline{11},\,
\underline{011}00000\underline{1},\,\underline{11}000\underline{1101}$ with the indicated spans.
The state labels are suppressed, and it suffices to keep in mind that the states at time zero appear in the
same ordering at the beginning and end of the trellis.
\begin{figure}[ht]
\centering
    \includegraphics[height=5.2cm]{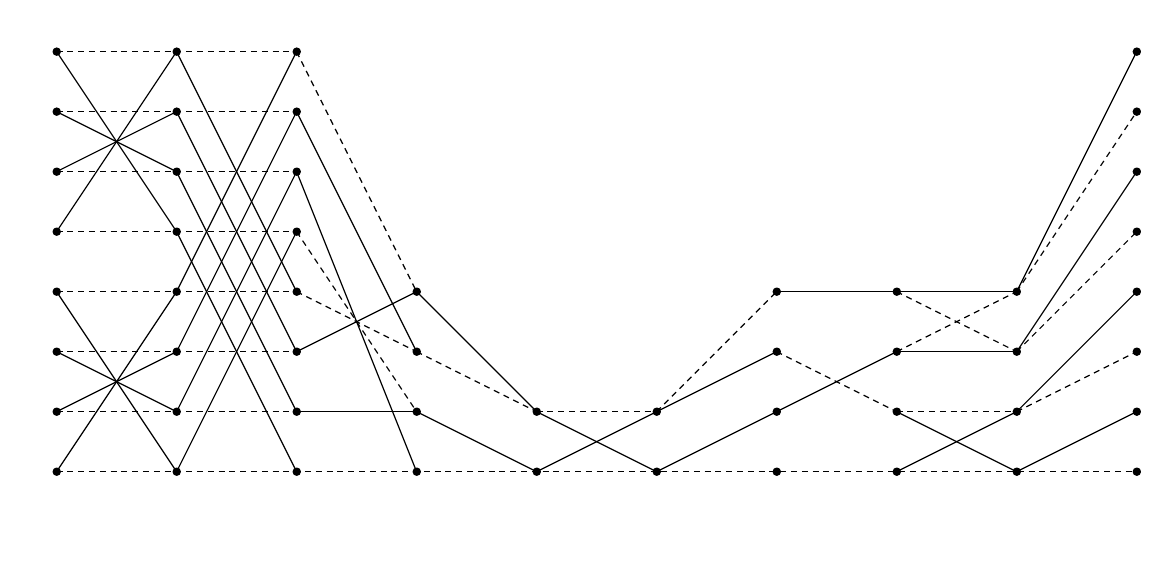}
    \\[-3ex]
    \caption{$2$-irreducible and $(m-3)$-unobservable trellis}
    \label{F-2O2C}
\end{figure}

\noindent The code~$\cC\subseteq\F_2^9$ generated by these vectors satisfies
$\chi(\cC)=3$ and $\chi(\cC^{\perp})=6$.
One can see directly that the trellis is $(m-2)$-observable, and it is also not hard to check that it is
$(m-2)$-controllable.
Thus, by Theorem~\ref{T-Charac-chi} the trellis is $2$-irreducible.
But the trellis is $(m-3)$-unobservable.
Even more, the unobservable valid path in the fragment $\cR^{[0,6)}$ satisfies Condition~A of
Theorem~\ref{T-ZeroRun}: there is no valid path from the ending state at time~$6$ to the zero state at time~$8$.
As a consequence, the trellis is reducible on the interval $[6,0)$.
Performing the reduction procedure as in the proof of Theorem~\ref{T-ZeroRun}, we obtain the
trellis shown in Figure~\ref{F-2O2C2} (where we added suitable state labels for further referencing).
\begin{figure}[ht]
\centering
    \includegraphics[height=4.4cm]{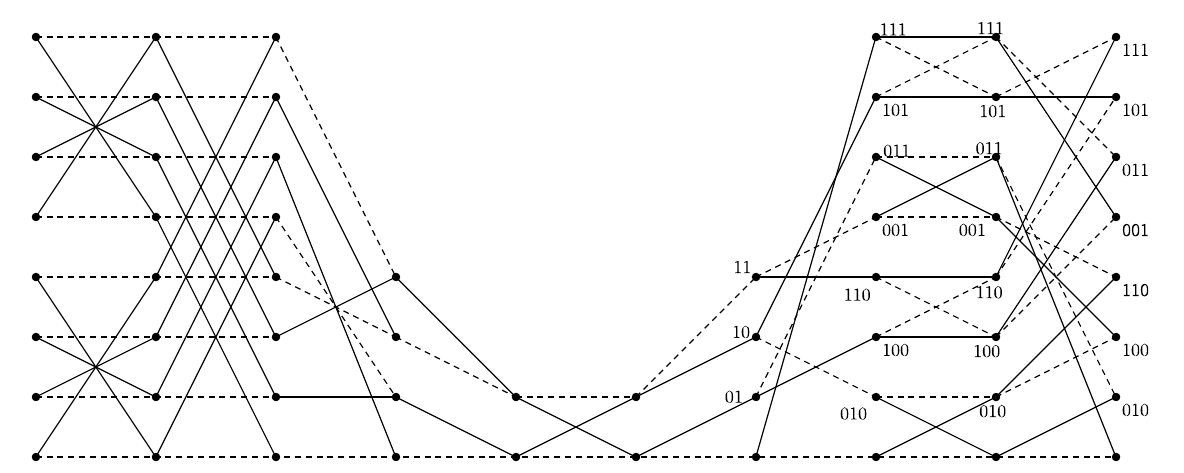}
    \\[-3ex]
    \caption{Step~1 of the reduction}
     \label{F-2O2C2}
\end{figure}

The (unique) subspace~$\cX$ of~$\cS_8^+$ as in Step~2 of the reduction process
is $\cX=\inner{010,111}\subseteq\cS_8^+$.
Trimming~$\cS_8^+$ to~$\cX$ leads to a trellis that is not trim at time~$6$ because there is no valid
path from~$11\in\cS_6$ to any state in~$\cX$.
Thus, as stated in Steps~3 and~4, we can subsequently trim the state spaces $\cS_7^+$ and~$\cS_6$.
In this case, trimming~$\cS_7^+$ to $\inner{010,111}$ and trimming~$\cS_6$ to $\inner{10}$
results in the trellis shown in Figure~\ref{F-2O2C3}.
\begin{figure}[ht]
\centering
    \includegraphics[height=4.4cm]{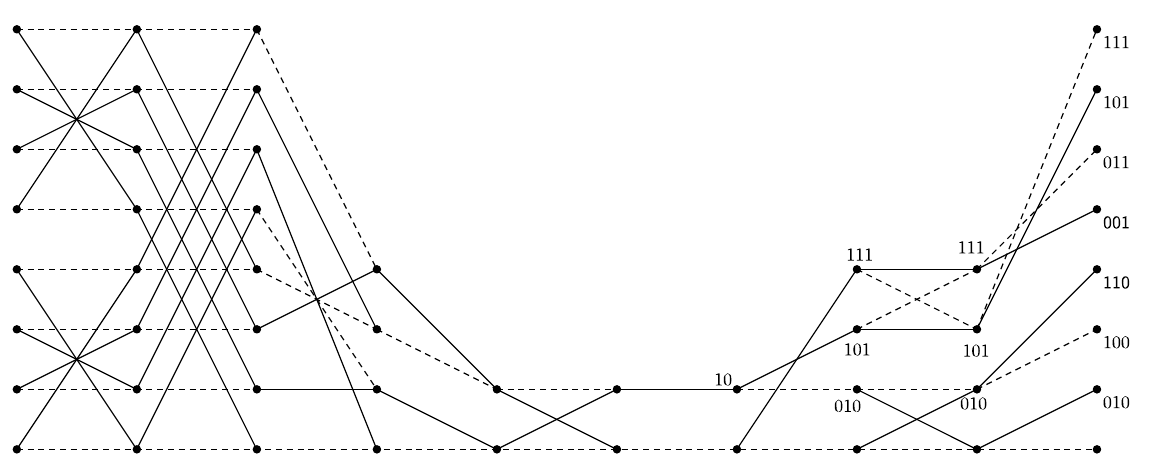}
    \\[-3ex]
    \caption{Steps~3 and~4 of the reduction leading to a strict $4$-reduction}
    \label{F-2O2C3}
\end{figure}
This trellis, denoted by~$\tilde{\cR}$, coincides with~$\cR$ on $[0,5)$ and is thus a strict and conservative
$4$-reduction.
The trellis is $(m-2)$-observable and $(m-2)$-controllable.
Furthermore, it can easily be checked that~$\tilde{\cR}$ and $\tilde{\cR}^{\circ}$ do not contain any unobservable
valid paths satisfying Conditions~A or~A$'$ of Theorem~\ref{T-ZeroRun}, and thus we do not have any other reduction method
at our disposal.
This is not surprising because it can be shown (with the aid of the class of KV-trellises as
introduced by Koetter and Vardy in~\cite{KoVa03}; see also Section~\ref{S-KV}), that~$\tilde{\cR}$ is a minimal trellis and
every trellis for the same code with state spaces of at most the same sizes is isomorphic
to~$\tilde{\cR}$.
Thus~$\tilde{\cR}$ is $t$-irreducible for every~$t$.
\QED
\end{exa}

\begin{exa}\label{E-SmallReduc}
The trellis in Fig.~\ref{F-TEMin1}(a) is the product trellis obtained from the generators
$\underline{1011}00$, $00\underline{1101},\,\underline{011}0\underline{11}$.
Its dual is shown in~(b).
\begin{figure}[!h]
\centering
    \includegraphics[height=2cm]{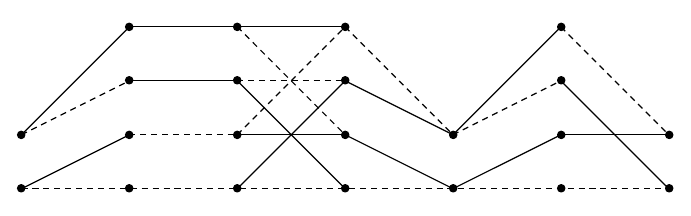}\quad \includegraphics[height=2cm]{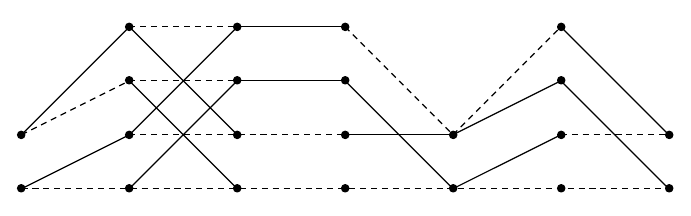}
    \\
    (a)\hspace*{7cm}(b)
    \\
    \caption{Dual pair of nonminimal trellises}
    \label{F-TEMin1}
\end{figure}

The code~$\cC$ generated by these vectors contains the word $100001$ and thus satisfies
$\chi(\cC)=2$, whereas $\chi(\cC^{\perp})=3$.
Hence Theorem~\ref{T-Charac-chi} is not applicable.
Furthermore, it is easy to see that the trellis and its dual do not contain any unobservable valid paths
satisfying Condition~A or~A$'$ from Theorem~\ref{T-ZeroRun}, and thus none of our reduction methods applies.
Yet, in this case the trellises are reducible.
Indeed, Fig.~\ref{F-TEMin3} shows mutually dual strict reductions on the interval $[4,3)$.
\begin{figure}[!h]
\centering
    \includegraphics[height=2cm]{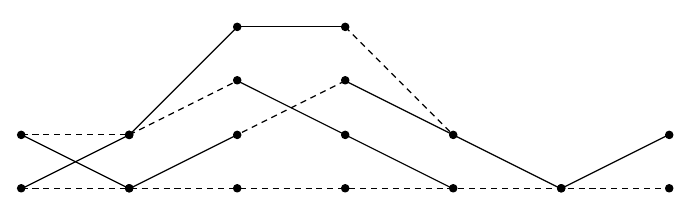}\quad \includegraphics[height=2cm]{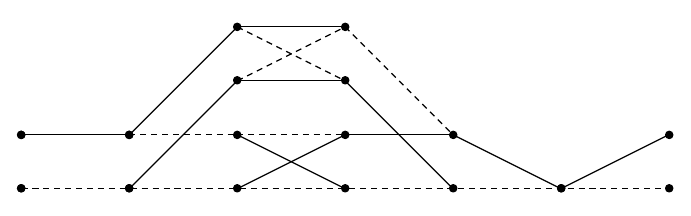}
    \\
    (a)\hspace*{7cm}(b)
    \\
    \caption{Dual pair of $(m\!-\!1)$-reductions}
    \label{F-TEMin3}
\end{figure}

These trellises are conventional, and thus minimal and $t$-irreducible for every~$t$.
Note that the trellis in Fig.~\ref{F-TEMin3}(a) is a conservative reduction of Fig.~\ref{F-TEMin1}(a), whereas
the dual reduction is a non-conservative reduction of Fig.~\ref{F-TEMin1}(b).

The trellises in Fig.~\ref{F-TEMin3} can be obtained constructively from those in Fig.~\ref{F-TEMin1}:
applying twice and in a suitable way the reduction from the proof of Theorem~\ref{T-ZeroRun} to Fig.~\ref{F-TEMin1}(a)
will result in Fig.~\ref{F-TEMin3}(a).
Even though Theorem~\ref{T-ZeroRun} does not apply, this will eventually lead to the reduction on the
interval~$[4,3)$.
With the aid of the class of KV-trellises, one can show that the trellis in Fig.~\ref{F-TEMin3}(a) is
the only trellis for~$\cC$ that is strictly smaller  than that in Fig.~\ref{F-TEMin1}(a).
Therefore, it is also the only possible reduction.
By duality, the same applies to the dual trellises in Fig.~\ref{F-TEMin3}(b) and Fig.~\ref{F-TEMin1}(b).
Hence we conclude that there exist strictly $t$-reducible trellises that do not allow a conservative reduction.
\QED
\end{exa}

As the last two examples have illustrated, it remains an open problem how to characterize $t$-irreducibi\-li\-ty for
$t\geq\ell:=\min\{\chi(\cC),\,\chi(\cC^{\perp})\}$.
By Theorem~\ref{T-jkObsIrr} we can expect $t$-reducibility only if there exists an unobservable
fragment of length $m-t$ in the trellis or its dual.
If such an unobservable valid path satisfies Condition~A or~A$'$ from Theorem~\ref{T-ZeroRun}, then the reduction
procedure from the proof of this theorem is applicable and the trellis is reducible.
As a consequence, the only remaining case is where the unobservable valid path does not satisfy either of
the conditions.
As shown in the proof of Theorem~\ref{T-Charac-chi} this means that the path lies on a valid trajectory (which,
for instance, is the case in any conventional trellis).
Obviously, the support of the associated codeword is contained in an interval of length~$t$.
Summarizing, one needs to study unobservable valid paths of length $\leq m-\ell$ that lie on valid trajectories.
Supported by many examples, we formulate the following conjecture.

\begin{conj}\label{Conj}
Any nonminimal trellis can be reduced constructively using a finite number of steps as in the reduction
procedure from the proof of
Theorem~\ref{T-ZeroRun} along with suitable state trimmings applied to the trellis or its dual.
\end{conj}

\section{Comparison to KV-Trellises}\label{S-KV}
In this section we relate our results to previous work on trellises and their complexity.
To this end, we restrict ourselves to trellises with symbol spaces~$\cA_i=\F$ for all~$i$ and to codes~$\cC$ in~$\F^m$
such that both~$\cC$ and~$\cC^{\perp}$ have full support, that is, the codewords do not all vanish at a fixed coordinate;
in other words, $\min\{\chi(\cC),\chi(\cC^{\perp})\}>1$.

Koetter and Vardy~\cite{KoVa03} showed that the search for possibly minimal trellises
can be narrowed to a certain canonical class, which we call \emph{KV-trellises}.
A KV-trellis is a product trellis based on $\dim \cC$ linearly independent generators with shortest spans that all
start and end at different positions.
A span $[\alpha,\alpha+r]$ (in the circular sense) is called a \emph{shortest span of the code} if~$r$ is the smallest
length of the spans of all codewords (in the sense of Definition~\ref{D-span}) that are nonzero at~$\alpha$;
see also~\cite{GLW11,GLW11a}.
KV-trellises may be regarded as the tail-biting version of the realizations resulting from the ``shortest basis'' approach in~\cite{Fo11a}.

Being product trellises, KV-trellises are state-trim and branch-trim, and from the choice of the generators it
follows that they are proper, observable and controllable.
Moreover, in \cite[Thm.~5.5]{KoVa03}  Koetter and Vardy have shown that each (reduced) minimal trellis is a KV-trellis.
However, the converse is not true: not all KV-trellises are minimal.

In \cite[Thm.~IV.3]{GLW11a} it is proven that the dual of a KV-trellis is a KV-trellis of the dual code.
As a consequence, Theorem~\ref{T-DualProd} implies that KV-trellises are $(m\!-\!1)$-observable and $(m\!-\!1)$-controllable,
and thus $1$-irreducible due to Theorem~\ref{T-1Red}.
With the machinery developed in~\cite{GLW11a} (more precisely, by a generalization of
Theorem~II.13 in~\cite{GLW11a}), one can show that if a code~$\cC$ satisfies $\min\{\chi(\cC),\chi(\cC^{\perp})\}>t$,
then all its KV-trellises (and their duals) are $(m-t)$-observable and hence $t$-irreducible
by Theorem~\ref{T-Charac-chi}.
However, this property does not characterize KV-trellises.
Indeed, the trellis in Fig.~\ref{F-2O2C} (and its dual) is not a KV-trellis\footnote{The sum of the first four generators given in Example~\ref{E-2Irred} results in a codeword with span $[5,8]$. Since this span is shorter than the span $[5,1]$ of the last
generator, the latter is not a shortest span of the code.},
but $2$-irreducible and represents a code satisfying $\min\{\chi(\cC),\chi(\cC^{\perp})\}=3$.

Summarizing, for a given~$t$ and a code~$\cC$ such that $\min\{\chi(\cC),\chi(\cC^{\perp})\}>t>1$
we have the following proper containments of trellis classes
\[
 \underbracket[0.5pt][7pt]{\bigg\{\!\textrm{minimal}\!\bigg\}\subsetneq\bigg\{\!\textrm{KV}\!\bigg\}\subsetneq
                            \bigg\{\!\!\begin{array}{l}\textrm{$t$-irreducible,}\\\textrm{T$_{\rm sb}$POC}\end{array}\!\!\!\bigg\}
   =\bigg\{\!\!\begin{array}{l}\textrm{$(m-t)$-obs./contr,}\\\textrm{T$_{\rm sb}$POC}\end{array}\!\!\!\bigg\}
       }_{\text{invariant under dualization}}\subsetneq
  \underbracket[0.5pt][7pt]{\bigg\{\textrm{NT$_{\rm sb}$POC}\bigg\}\subsetneq\bigg\{\textrm{T$_{\rm sb}$POC}\bigg\}}_{\text{not invariant under dualization}},
\]
where N stands for nonmergeable, T$_{\rm sb}$ for state- and branch-trim (as opposed to the weaker trimness),
and, as before, P,O,C stand for proper, observable, and controllable, respectively.

As indicated, the last two containments are strict as well:
for instance, the trellis in Figure~\ref{F-MergTrellis}(a) is in the rightmost class, but not in $\{\textrm{NT$_{\rm sb}$POC}\}$;
moreover, it is straightforward to verify that the product trellis generated by
$00\underline{101}00,\,000\underline{1011}$, $\underline{1}000\underline{100},\,\underline{011}000\underline{1}$
realizes a code~$\cC$ that satisfies $\min\{\chi(\cC),\chi(\cC^{\perp})\}=3$ and
is in $\{\textrm{NT$_{\rm sb}$POC}\}$, but $2$-reducible.

The three leftmost classes are invariant under taking duals.
This is clear for minimal trellises and for $t$-irreducible trellises and follows from \cite[Thm.~IV.3]{GLW11a} for KV-trellises.
Fig.~\ref{F-BCJR} and Fig.~\ref{F-MergTrellis} show that the two rightmost classes are not invariant under taking duals.

\section{Conclusion}
We have presented constructive procedures for reducing a given tail-biting trellis realization and its dual and have provided criteria for
when a trellis is irreducible on an interval of length~$t$.
The criteria are sufficient for all codes, and necessary and sufficient for codes of minimum span length
bigger than~$t$.
We have also discussed the remaining case of reducibility on intervals of length at least the minimum span length of the code.
While we believe that all nonminimal trellises can be reduced, finding a constructive procedure remains a largely open problem.

As a main tool of our approach, we have used trellis fragments, i.e., realizations obtained by cutting two edges in
the normal graph of the tail-biting trellis.
We have introduced the notions of fragment controllability and observability, which are, naturally, the same as those in classical
linear systems theory, and have shown that they are mutually dual.

With the aid of fragment trimness, which implies that every valid path in the fragment is the restriction of a
valid trajectory in the entire trellis, we have presented criteria for state-trimness and branch-trimness of a tail-biting trellis.
Using the well-known fact that a tail-biting trellis is a product trellis if and only if it is state-trim and branch-trim,
our results have also led to a characterization of when the dual of a product trellis is a product trellis.

Finally, we have discussed the relation of our results to the prior tail-biting trellis literature that relies on product
representations.

Beyond trellises, we believe that many of our results can be generalized to normal realizations on general graphs.

\appendix
\section{Controllability and Connectedness}\label{S-App1}
In this appendix we discuss the relationship between controllability of a trellis realization and connectedness of
the trajectories in its trellis diagram.
In \cite[Thm.~10]{FGL12} it was shown that if a trellis is trim but not controllable, then its
valid trajectories partition into disconnected subsets.
Moreover, it was noted in~\cite{FGL12} that, using the product representation of~\cite{KoVa03}, a state-trim and
branch-trim trellis is uncontrollable if and only if its trajectories are disconnected.
We will now show that this statement holds if the trellis is merely state-trim.

\begin{theo}\label{T-ContrConn}
A state-trim trellis is controllable if and only if it is connected.
\end{theo}
\begin{proof}
The if-part has been proven (for trim trellises) in \cite[Thm.~10]{FGL12}.
For state-trim and branch-trim trellises a proof of the converse is sketched in~\cite[Sec.~IV.F]{FGL12} by
using the fact that every such trellis is a product realization of one-dimensional trellises.
For non-branch-trim trellises a proof of the converse is as follows.
Let~$\cR$ be a controllable, state-trim trellis and suppose~$\cR$ is not connected.
Consider the connected component of~$\cR$ containing the zero trajectory.
Denote the state and constraint sets of this subtrellis by~$\cS_{i,0}$ and $\cC_{i,0}$, respectively.
With the aid of state-trimness one easily verifies that~$\cC_{i,0}$ is a linear subspace of~$\cC_i$
for each~$i$.
Notice that $\cS_{i,0}$ is the union of the projections of~$\cC_{i,0}$ and $\cC_{i-1,0}$ on~$\cS_i$.
Again using state-trimness one can see that $\cS_{i,0}$ is a linear subspace of~$\cS_i$ for all~$i$.
Thus, the connected component forms a linear subtrellis~$\cR_0$ with state spaces~$\cS_{i,0}$ and
constraint codes~$\cC_{i,0}$. Denote its behavior by~$\Bf_0$.
Then the behavior~$\Bf$ of~$\cR$ is the union of~$q^\ell$ disconnected subbehaviors, where $\ell:=\dim\Bf-\dim\Bf_0$,
and by linearity and state-trimness each state space~$\cS_i$ and constraint code~$\cC_i$ of~$\cR$ is the union of
$q^\ell$ cosets of~$\cS_{i,0}$ and $\cC_{i,0}$, respectively.
Using controllability of~$\cR$ we obtain
$\dim\Bf_0<\dim\Bf=\sum_i(\dim\cC_i-\dim\cS_i)=\sum_i(\dim\cC_{i,0}-\dim\cS_{i,0})$, and this contradicts
Theorem~\ref{T-ContrTest} for the trellis~$\cR_0$.
\end{proof}

We wish to point out that state-trimness is indeed necessary for the only-if part to be true.
The two linear, non-state-trim trellises shown in Fig.~\ref{F-DisConnContr} are disconnected, yet form controllable
realizations of the code $\cC=\{00\}$.
It is also worth observing that in the trellis in Fig.~\ref{F-DisConnContr}(b) the connected component containing the zero
trajectory is not a linear subtrellis.
\begin{figure}[ht]
\centering
    \includegraphics[height=2.5cm]{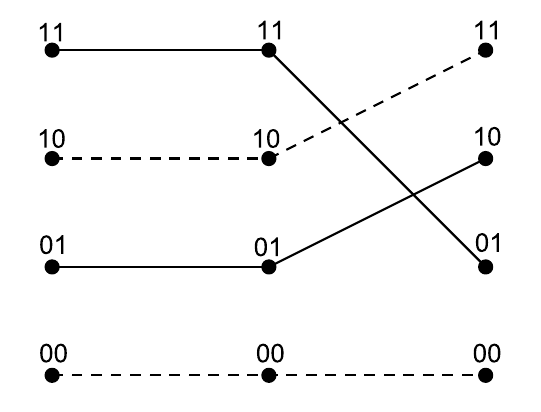}\qquad\qquad \includegraphics[height=1.7cm]{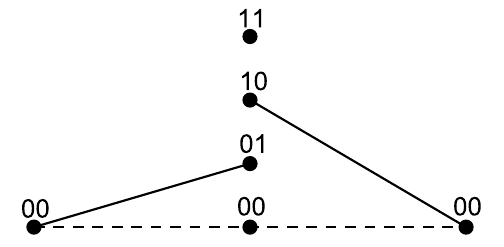}
    \\
    \mbox{}\hspace{-.5cm}(a) \hspace*{4.5cm} (b)
    \caption{Disconnected and controllable trellises}
    \label{F-DisConnContr}
\end{figure}

Finally, we note that~\cite[Sec.~IV.G]{FGL12} shows that Theorem~\ref{T-ContrConn} does not generalize to
normal linear realizations on general graphs with cycles.

\section{Proof of Theorem~\ref{T-ZeroRun}}\label{S-AppA}
Without loss of generality we assume that~$\cR$ satisfies Condition~A.
In a first step we expand the trellis~$\cR$ to an unobservable trellis by increasing
the state spaces at times~$m-t+1,\ldots,m-1$ by one dimension.
Thus, pick new states $\tilde{s}_i\not\in\cS_i$ for $i=m-t+1,\ldots,m-1$.
For ease of notation, define $\tilde{s}_{m-t}:=s_{m-t},\,\tilde{s}_0:=s_0$, and expand
the fragment $\cR^{[m-t,0)}$ via
\begin{equation}\label{e-AexpT}
  \cS^+_i = \cS_i \oplus \langle\tilde{s}_i\rangle\text{ for } i \in (m-t,0)\ \text{ and }\
  \cC^+_i = \cC_i \oplus \langle (\tilde{s}_i, 0, \tilde{s}_{i+1})\rangle\text{ for } i \in [m-t,0).
\end{equation}
By construction, the expanded trellis~$\cR^+$ has a nontrivial unobservable valid trajectory $(\zerob,\tilde{\sb})$ that passes
through the states~$\tilde{s}_i,\,i\in [m-t,0]$.

Using the fact that $\tilde{s}_i\not\in\cS_i$ for $i\in(m-t,0)$,  we obtain immediately for $i\in\{m-t+1,\ldots,m-2\}$
\begin{equation}\label{e-branches}
 (v_i+\alpha\tilde{s}_i,a_i,v_{i+1}+\beta\tilde{s}_{i+1})\in\cC_i^+
 \text{ for some } (v_i,a_i,v_{i+1})\in\cS_i\times\cA_i\times\cS_{i+1},\ \alpha,\beta\in\F\Longrightarrow  \alpha=\beta.
\end{equation}
The trellis~$\cR^+$ has the following properties.

1) The behavior of~$\cR^+$ is given by $\Bf^+=\Bf\oplus\inner{(\zerob,\tilde{\sb})}$, and thus~$\cR^+$ represents
the same code~$\cC$.
To see this, consider a valid trajectory in~$\Bf^+$.
Due to~\eqref{e-branches}, the state sequence must be of the form
$(v_0,\ldots,v_{m-t},\,v_{m-t+1}+\alpha\tilde{s}_{m-t+1},\ldots,v_{m-1}+\alpha\tilde{s}_{m-1},v_0)$,
where $v_i\in\cS_i$ for all~$i$.
Subtracting the unobservable trajectory $\alpha(\zerob,\tilde{\sb})$ yields a valid trajectory that is entirely
in the subtrellis~$\cR$, hence it is an element of~$\Bf$.

2) There is no valid path from~$\tilde{s}_{m-t}\in\cS_{m-t}$ to~$0\in\cS_{m-1}^+$ in the fragment $(\cR^+)^{[m-t,m-1)}$.
To show this, suppose we have such a path.
By~\eqref{e-branches} its state sequence is of the form
$(\tilde{s}_{m-t},v_{m-t+1}+\alpha\tilde{s}_{m-t+1},\ldots,v_{m-2}+\alpha\tilde{s}_{m-2},0)$, where $v_i\in\cS_i$, and
once more by~\eqref{e-branches} we conclude that $\alpha=0$.
But then the given path is a valid path in~$\cR^{[m-t,m-1)}$, and this contradicts Condition~A.

3) There exists a subspace $\cX$ of~$\cS_{m-1}^+$ satisfying
$\cX\oplus\inner{\tilde{s}_{m-1}}=\cS_{m-1}^+$ and such that there is no valid path in~$(\cR^+)^{[m-t,m-1)}$
from~$\tilde{s}_{m-t}$ to any $x\in\cX$.
This can be seen as follows.
Put
\[
    \cY:=\{s\in\cS_{m-1}^+\mid \text{there is a valid path from~$0\in\cS_{m-t}$ to~$s$ in the fragment $(\cR^+)^{[m-t,m-1)}$}\}.
\]
Then $\tilde{s}_{m-1}\not\in\cY$ because if there was a valid path from~$0$ to~$\tilde{s}_{m-1}$, then the existence
of a path with state sequence $(\tilde{s}_{m-t},\tilde{s}_{m-t+1},\tilde{s}_{m-t+2},\ldots,\tilde{s}_{m-1})$ leads to a valid path
in~$(\cR^+)^{[m-t,m-1)}$ from~$\tilde{s}_{m-t}$ to $0\in\cS_{m-1}^+$, and this contradicts our observation in~2).
Using again the path from~$\tilde{s}_{m-t}$ to~$\tilde{s}_{m-1}$ we observe that for each $\alpha\in\F$ the coset
$\cY+\alpha\tilde{s}_{m-1}$ is exactly the set of all states in~$\cS_{m-1}^+$ that can be reached by a valid path from
$\alpha\tilde{s}_{m-t}$.
Now, we may choose any subspace $\cZ\subset\cS_{m-1}^+$ such that $\cZ\oplus\cY\oplus\inner{\tilde{s}_{m-1}}=\cS_{m-1}^+$
and put $\cX=\cZ\oplus\cY$.

Having established these properties we can perform the reduction:

a) Trim~$\cS_{m-1}^+$ to the subspace~$\tilde{\cS}_{m-1}:=\cX$, where~$\cX$ is as in~3).
By Remark~\ref{R-TrimUnObs} this results in a trellis that still represents~$\cC$.
Denote its constraint code at time~$m-2$ by $\tilde{\cC}_{m-2}$.

b) Let $\tilde{\cS}_{m-2}$ be the projection of~$\tilde{\cC}_{m-2}$ on the state space~$\cS_{m-2}^+$.
Then~$\tilde{\cS}_{m-2}$ is contained in the set
$\{s\in\cS_{m-2}^+\mid \text{there is no valid path from~$s_{m-t}$ to~$s$}\}$, and
the latter is a proper subset of~$\cS_{m-2}^+$.
Thus, $\dim\tilde{\cS}_{m-2}<\dim\cS_{m-2}^+$.
Obviously, the states not in~$\tilde{\cS}_{m-2}$ are not on any valid trajectory, and thus we may trim~$\cS_{m-2}^+$
to~$\tilde{\cS}_{m-2}$.
After this trimming denote the constraint code at time~$m-3$ by $\tilde{\cC}_{m-3}$ and continue in the same manner.

c) All this shows that we can trim all state spaces~$\cS^+_{m-1},\,\cS^+_{m-2},\ldots,\cS^+_{m-t+1}$ by  one dimension.
Since the branches that have been added in~\eqref{e-AexpT} will be trimmed, this also reduces the constraint
codes~$\cC^+_i,\,i=m-t,\ldots,m-1$, by one dimension.
Thus the resulting trellis, denoted by~$\tilde{\cR}$, is an $[m-t,0)$-reduction with the same state space and
constraint code dimensions as~$\cR$.
Consequently, the same is true for the dual reduction~$\tilde{\cR}^{\circ}$ of~$\cR^{\circ}$.
Finally, by construction,~$\tilde{s}_{m-t}\in\cS_{m-t}$ is not on any branch in the constraint code~$\tilde{\cC}_{m-t}$,
and thus~$\tilde{\cR}$ is not trim at time~$m-t$.
Trimming results in the desired strict and conservative $[m-t-1,0)$-reduction of~$\cR$, whose dual is a
reduction of~$\cR^{\circ}$ of the same type. \mbox{}\hfill$\Box$

\bibliographystyle{abbrv}
\bibliography{literatureAK,literatureLZ}

\begin{thebibliography}{10}

\bibitem{AlMao11}
A.~Al-Bashabsheh and Y.~Mao.
\newblock Normal factor graphs and holographic transformations.
\newblock {\em IEEE Trans. Inform. Theory}, IT-57:752--763, 2011.

\bibitem{CFV99}
A.~R. Calderbank, {G.~D. Forney, Jr.}, and A.~Vardy.
\newblock Minimal tail-biting trellises: {T}he {G}olay code and more.
\newblock {\em IEEE Trans. Inform. Theory}, IT-45:1435--1455, 1999.

\bibitem{Con12}
D.~Conti.
\newblock {\em An Algebraic Development of Trellis Theory}.
\newblock PhD thesis, University College Dublin, 2012.

\bibitem{Fo01}
{G.~D. Forney, Jr.}
\newblock Codes on graphs: {N}ormal realizations.
\newblock {\em IEEE Trans. Inform. Theory}, IT-47:520--548, 2001.

\bibitem{Fo11}
{G.~D. Forney, Jr.}
\newblock Codes on graphs: {D}uality and {M}ac{W}illiams identities.
\newblock {\em IEEE Trans. Inform. Theory}, IT-57:1382--1397, 2011.

\bibitem{Fo11a}
{G.~D. Forney, Jr.}
\newblock Minimal realizations of linear systems: {T}he ``shortest basis''
  approach.
\newblock {\em IEEE Trans. Inform. Theory}, IT-57:726--737, 2011.

\bibitem{FGL12}
{G.~D. Forney, Jr.} and H.~Gluesing-Luerssen.
\newblock Codes on graphs: {O}bservability, controllability and local
  reducibility.
\newblock Preprint 2012. Accepted for publication in {\em IEEE Trans.\ Inform.\
  Theory}. ArXiv: cs.IT/1203.3115v2.

\bibitem{GLW11a}
H.~Gluesing-Luerssen and E.~Weaver.
\newblock Characteristic generators and dualization for tail-biting trellises.
\newblock {\em IEEE Trans. Inform. Theory}, IT-57:7418--7430, 2011.

\bibitem{GLW11}
H.~Gluesing-Luerssen and E.~Weaver.
\newblock Linear tail-biting trellises: {C}haracteristic generators and the
  {BCJR}-construction.
\newblock {\em IEEE Trans. Inform. Theory}, IT-57:738--751, 2011.

\bibitem{Ka09}
N.~Kashyap.
\newblock On minimal tree realizations of linear codes.
\newblock {\em IEEE Trans. Inform. Theory}, IT-55:3501--3519, 2009.

\bibitem{Koe02}
R.~Koetter.
\newblock On the representation of codes in {F}orney graphs.
\newblock In {\em Codes, Graphs, and Systems: {A} celebration of the life and
  career of {{G}.~{D}avid {F}orney, Jr.} {\rm (R.~E. Blahut and R. Koetter,
  eds.)}}, pages 425--450. Kluwer Academic Publishers, 2002.

\bibitem{KoVa02}
R.~Koetter and A.~Vardy.
\newblock On the theory of linear trellises.
\newblock In {\em Information, Coding and Mathematics {\rm (M.~Blaum and
  P.~G.~Farrell and H.~C.~A.~van~Tilborg, eds.)}}, pages 323--354. Kluwer,
  Boston, 2002.

\bibitem{KoVa03}
R.~Koetter and A.~Vardy.
\newblock The structure of tail-biting trellises: {M}inimality and basic
  principles.
\newblock {\em IEEE Trans. Inform. Theory}, IT-49:2081--2105, 2003.

\bibitem{KschSo95}
F.~R. Kschischang and V.~Sorokine.
\newblock On the trellis structure of block codes.
\newblock {\em IEEE Trans. Inform. Theory}, IT-41:1924--1937, 1995.

\bibitem{LiSh00}
S.~Lin and R.~Y. Shao.
\newblock General structure and construction of tail-biting trellises for
  linear block codes.
\newblock In {\em Proceedings of the 2000 IEEE International Symposium on
  Information Theory}, page 117, 2000.

\bibitem{MaoKsch05}
Y.~Mao and F.~R. Kschischang.
\newblock On factor graphs and the {F}ourier transform.
\newblock {\em IEEE Trans. Inform. Theory}, IT-51:1635--1649, 2005.

\bibitem{NoSh06}
A.~V. Nori and P.~Shankar.
\newblock Unifying views of tail-biting trellis constructions for linear block
  codes.
\newblock {\em IEEE Trans. Inform. Theory}, IT-52:4431--4443, 2006.

\bibitem{ShBe00}
Y.~Shany and Y.~Be'ery.
\newblock Linear tail-biting trellises, the square root bound, and applications
  for {R}eed-{M}uller codes.
\newblock {\em IEEE Trans. Inform. Theory}, IT-46:1514--1523, 2000.

\bibitem{Va98}
A.~Vardy.
\newblock Trellis structure of codes.
\newblock In {\em Handbook of Coding Theory, Vol.~2 {\rm (V.~S.~Pless and
  W.~C.~Huffman, eds.)}}, pages 1989--2117. Elsevier, Amsterdam, 1998.

\bibitem{VaKsch96}
A.~Vardy and F.~R. Kschischang.
\newblock Proof of a conjecture of {M}c{E}liece regrading the expansion index
  of the minimal trellis.
\newblock {\em IEEE Trans. Inform. Theory}, IT-42:2027--2034, 1996.

\end{thebibliography}

\end{document}